\documentclass[journal]{IEEEtran}

\usepackage{cite}

\usepackage{amsmath}
\usepackage{amsfonts}
\usepackage{amssymb}
\usepackage{graphicx}
\usepackage{verbatim}
\usepackage{enumitem}
\usepackage[dvipsnames]{xcolor}
\usepackage{colortbl}
\usepackage{bm}
\usepackage{hyperref}
\hypersetup{colorlinks=true}
\usepackage{eurosym}
\usepackage[colorinlistoftodos]{todonotes}
\usepackage{bbold}

\usepackage[framemethod=tikz]{mdframed}
\usepackage{amsthm}
\usetikzlibrary{arrows,intersections,chains}
\usetikzlibrary{positioning}
\usepackage{subcaption}
\usepackage{verbatim}
\usepackage{varwidth}
\usepackage{pgfplots}
\pgfplotsset{compat=newest}
\usepackage[framemethod=tikz]{mdframed}
\usetikzlibrary{arrows,intersections,chains,calc}
\usetikzlibrary{positioning}
\usetikzlibrary{shapes}

\usetikzlibrary{shapes.geometric,backgrounds}

\pgfplotsset{compat=1.11,
	/pgfplots/ybar legend/.style={
		/pgfplots/legend image code/.code={%
			\draw[##1,/tikz/.cd,yshift=-0.25em]
			(0cm,0cm) rectangle (3pt,0.8em);},
	},
}
\tikzset{
	basic box/.style = {
		shape = rectangle,
		align = center,
		draw  = #1,
		fill  = #1!25,
		rounded corners},
	header node/.style = {
		font          = \strut\Large,
		text depth    = +0pt,
		fill          = white,
		draw},
	header/.style = {%
		inner ysep = +1.5em,
		append after command = {
			\pgfextra{\let\TikZlastnode\tikzlastnode}
			node [header node] (header-\TikZlastnode) at (\TikZlastnode.north) {#1}
			node [span = (\TikZlastnode)(header-\TikZlastnode)]
			at (fit bounding box) (h-\TikZlastnode) {}
		}
	},
}

\usepackage{array}
\usepackage{tabularx,booktabs}
\newcolumntype{C}{>{\centering\arraybackslash}X} % centered version of "X" type
\newtheorem{proposition}{Proposition}
\newtheorem{theorem}{Theorem}

\newtheorem{definition}{Definition}

\newtheorem{remark}{Remark}[section]
\newtheorem{corollary}{Corollary}[section]

\newcommand{\bd}{\boldsymbol}
\newcommand{\mc}{\mathcal}

\begin{document}

% \author{\IEEEauthorblockN{Anibal Sanjab,\IEEEauthorrefmark{} \and 
% Hélène Le Cadre,\IEEEauthorrefmark{} and Yuting Mou~\IEEEmembership{}\\}
% \IEEEauthorblockA{\small\IEEEauthorrefmark{}VITO/EnergyVille, Thor Park 8310, 3600 Genk, Belgium}\\
% %\IEEEauthorblockA{\IEEEauthorrefmark{2}Twentieth Century Fox, Springfield, USA}
% %\IEEEauthorblockA{\IEEEauthorrefmark{3}Starfleet Academy, San Francisco, CA 96678 USA}
% %\IEEEauthorblockA{\IEEEauthorrefmark{4}Tyrell Inc., 123 Replicant Street, Los Angeles, CA 90210 USA}% <-this % stops an unwanted space
% %\thanks{Manuscript received December 1, 2012; revised August 26, 2015. 
% {\small %Corresponding author: H. Le Cadre (email: helene.lecadre@vito.be).
% }}

\title{{\huge TSO-DSOs Stable Cost Allocation for the Joint Procurement of Flexibility: A Cooperative Game Approach}\vspace{-0.1cm}}

\author{Anibal Sanjab, 
	Hélène~Le~Cadre, and~Yuting~Mou% <-this % stops a space
	\thanks{The authors are with the Flemish Institute for Technological Research VITO/EnergyVille, Thor Park 8310, 3600 Genk, Belgium. e-mails \{anibal.sanjab, helene.lecadre, yuting.mou\}@vito.be. The authors have equally contributed to this work.}
	\thanks{This work is supported by the EU's Horizon 2020 research and innovation programme under grant agreement No 824414 -- {\bf CoordiNet} project.}\vspace{-0.6cm}% <-this % stops a space
}

\maketitle

\begin{abstract}
In this paper, a transmission-distribution systems flexibility market is introduced, in which system operators (SOs) jointly procure flexibility from different systems to meet their needs (balancing and congestion management) using a common market. This common market is, then, formulated as a cooperative game aiming at identifying a stable and efficient split of costs of the jointly procured flexibility among the participating SOs to incentivize their cooperation. The non-emptiness of the core of this game is then mathematically proven, implying the stability of the game and the naturally-arising incentive for cooperation among the SOs. Several cost allocation mechanisms are then introduced, while characterizing their mathematical properties. Numerical results focusing on an interconnected system (composed of the IEEE 14-bus transmission system and the Matpower 18-bus, 69-bus, and 141-bus distributions systems) showcase the cooperation-induced reduction in system-wide flexibility procurement costs, and identifies the varying costs borne by different SOs under various cost allocations methods.\vspace{-0.4cm} 
\end{abstract}

%\begin{IEEEkeywords}
%test
%\end{IEEEkeywords}

\section{Introduction}

% TSO-DSO coordination schemes

The increasing integration of distributed energy resources (DERs) and electrification of the consumer energy space (e.g., transportation and heating) pose challenges for grid operation, due to the induced uncertainty and changing load patterns. However, this new energy landscape also enables an unprecedented growing volume of invaluable flexibility\footnote{Flexibility is the ability to dynamically modify consumption and generation patterns providing, as a result, a service to system operators.} (from different voltage levels of the grid) thereby providing essential services (e.g., congestion management and balancing) for transmission system operators (TSOs) and distribution system operators (DSOs). In this respect, the introduction of market mechanisms for the procurement of flexibility from flexibility services provides (FSPs) has been increasingly recommended in policies~\cite{Ceer}, and has been the center of several recent works in the literature~\cite{DLMP_SOCP_HLC,TSODSOCooperation,TSODSOCoordination_UCL,CIRED21,JointTSODSO_Roos,ReboundOffers} and demonstration projects~\cite{FlexmarketsProjects_FSR}.
%for different system services for TSOs and DSOs is being increasingly recommended in policies~\cite{Ceer}, while recent years have also shown increasing analysis of flexibility markets in the literature~\cite{DLMP_SOCP,DLMP_SOCP_HLC,TSODSOCooperation,TSODSOCoordination_UCL,CIRED21,JointTSODSO_Roos,FlexMarkets_Review,ReboundOffers} and in various demonstration projects~\cite{FlexmarketsProjects_FSR}.

As FSPs could provide their flexibility as a service to different system operators (SOs), a major branch of the literature has focused on the SOs' joint procurement (i.e. co-optimization) of flexibility~\cite{JointTSODSO_Roos,TSODSOCooperation}, to maximize the grid and system-level value of this flexibility. In particular, a key focus has been shed on the need for coordination between SOs to achieve joint procurement, not only for optimization of economic efficiency but also to ensure that the activated flexibility meets grid operational constraints of all the grids involved~\cite{DLMP_SOCP_HLC,TSODSOCooperation,CIRED21,TSODSOCoordination_VITO,TSODSOCoordination_UCL,TDLocalMarkets}. However, \emph{when jointly procuring flexibility, it is paramount to decide on how the costs of this flexibility should be divided among the participating SOs in the most efficient and fair way, and most importantly, in a manner that incentivizes system operators to naturally collaborate and jointly procure flexibility rather than running their disjoint markets}. To the best of our knowledge\footnote{We note that the work in~\cite{TSODSOCooperation} looks at a Retailer-DSO-TSO coordination setting while focusing particularly on the Shapley value for the numerical computation of cost allocation, hence, defers in scope from the current work.}, no work in the literature has presented a fundamental analysis and comprehensive solution to this joint cooperation and cost allocation problem, as is the goal of the current work. 

In this paper, we develop a novel cooperative game-theoretic\footnote{Recent years have seen a few applications of cooperative game concepts to energy system problems, with primary focus on the residential sector and shared energy investments, such as in~\cite{hupez,robu,PowerSystemGT}.} approach in which 1) we develop the concept of cooperation stability between the TSO and DSOs in a joint flexibility market, 2) determine whether cooperation between the TSO and DSOs is naturally incentivized and the conditions therefore, and 3) derive several cost allocation mechanisms and analytically investigate their properties focusing on concepts such as stability, efficiency, and computational complexity.

Towards this end, we first introduce a novel flexibility market model including a TSO and multiple DSOs for jointly procuring congestion management and balancing services while explicitly accounting for grid constraints. This framework is developed by first introducing disjoint TSO and DSO level markets and joining them in a common market setting. We then formulate the joint procurement of flexibility between any subset of SOs as a cooperative (cost allocation) game and analyze its properties. We subsequently prove the concavity of this game, hence, guaranteeing the non-emptiness of its core\footnote{The core is a game-theoretic concept used to assess cooperation stability, and will be formally defined in Section~\ref{sec:CostGame}.} under several pricing mechanisms. This, as a result, proves that 
%adequate cost allocation mechanisms would guarantee that cooperation between SOs could naturally occur as it is beneficial to all operators. In other words, 
\emph{all SOs would be better off joining the grand coalition (i.e., the set of all SOs in a common market setting) as compared to joining any sub-coalition (i.e., a coalition of a subset of SOs, or forming disjoint markets with no cooperation)}. Then, we introduce several cost allocation mechanisms to allocate the costs of the jointly procured flexibility among the different SOs (namely, \emph{the Shapley value, normalized Banzhaf index, cost gap allocation, Lagrangian-based cost allocation, equal profit method, and proportional cost}) and investigate their stability (whether they belong to the core), allocation adequacy, and other mathematical properties such as: %, and allocation adequacy. 
%for allocating the flexibility procurement costs among the different SOs. 
%The properties investigated focus on 
efficiency, symmetry, additivity, dummy player, anonymity, and computational complexity. 

Our analytical conclusions and results are further corroborated using a flexibility market case analysis focusing on an interconnected test system composed of the IEEE 14-bus transmission system interconnected with three distribution systems -- namely, the Matpower 141-bus, 69-bus and 18-bus systems. 
%In this case analysis, we consider different SOs aiming at meeting their flexibility needs using upward and downward flexibility bids submitted to the flexibility market from different transmission and distribution system nodes, and we compare the costs associated with each SO under the different cost allocation mechanisms and for different limits on the interface flow between the TSO and different DSOs. 
The numerical results highlight 1) the significant reduction in system cost when more SOs joint the common market grand coalition, highlighting the benefits of cooperation, 2) the essential role of higher interface flows on reaping the benefits of cooperation, and 3) the disparity that can result from different cost allocation methods, hence, providing indispensable inputs to SOs, regulators, and decision makers.

The rest of the paper is organized as follows. Section~\ref{sec:SystemModel} presents the disjoint and common markets formulations. Section~\ref{sec:CostGame} introduces the cooperative game formulation. Section~\ref{sec:Stability} investigates the stability of the game, while Section~\ref{sec:allocation_mechanisms} provides different cost allocation mechanisms and characterizes their properties. Section~\ref{sec:NumericalResults} introduces the numerical results, and Section~\ref{sec:Conclusion} concludes the paper.%\vspace{-0.2cm} 

\section{Systems and Markets Models}\label{sec:SystemModel}
Consider a transmission system (operated by a TSO) composed of a set of nodes, $\mathcal{N}^T$, and set of lines, $\mathcal{L}^T$, represented by a graph $\mathcal{G}^T(\mathcal{N}^T,\mathcal{L}^T)$. At a subset, $\mathcal{N}^D$, of these nodes, distribution systems (each operated by a DSO) are connected. 
%, as showcased in Fig.~\ref{fig:TSO_DSO_example}. 
We refer to $\mathcal{N}^D\subseteq\mathcal{N}^T$, as the set of interface nodes, containing $\mathcal{N}^D$ nodes from each of which stems one of the $N^D$ different distribution systems. We let a) $p_n^T$ be the net real power injection at node $n\in\mathcal{N}^T$, b) $\boldsymbol{p}^{T,o}$ and $\boldsymbol{d}^{T,o}$ be, respectively, the vectors of anticipated base injection and load at all transmission system nodes $\mathcal{N}^T$, c) $P_{ij}^T$ denote the real power flow over line $\{i,j\}\in\mathcal{L}^T$, with maximum thermal line limit $F^{T,\textrm{max}}_{ij}$, d) $T_{n^T}^p$ and $T_{n^T}^q$ denote the active and reactive power transfer to the distribution system connected to node $n^T\in\mathcal{N}^D$, and e) $X_{(i,j),n}$ denote the generation shift factor of $P_{ij}^T$ with respect to the net injection at node $n\in\mathcal{N}^T$, capturing the change in the flow over line $\{i,j\}\in\mathcal{L}^T$ due to a change in net injection at node $n\in\mathcal{N}^T$.

We denote a distribution system connected to a transmission node $n^T\in\mathcal{N}^D$ by DSO--$n^T$. Each DSO--$n^T$ is composed of a radial distribution network composed of $\mathcal{N}^{n^T}$ nodes (where node $n_0^{n^T} \in \mathcal{N}^{n^T}$ is the root node of DSO--$n^T$) and $\mathcal{L}^{n^T}$ distribution lines forming a graph, $\mathcal{G}^{n^T}(\mathcal{N}^{n^T},\mathcal{L}^{n^T})$. In each DSO--$n^T$, we let $A(n)$ denote the ancestor node of $n\in\mathcal{N}^{n^T}\setminus\{n_0^{n^T}\}$ and $\mathcal{K}(n)$ the set of predecessor nodes of $n$. For the root node $n_0^{n^T}$, $A(n_0^{n^T})=n^T$. 
For each DSO--$n^T$, we let a) $s_n^{n^T}=p_n^{n^T}+jq_n^{n^T}$ be the net complex power injection at bus $n\in\mathcal{N}^{n^T}$, b) $\boldsymbol{p}^{n^T\!,o}$ and $\boldsymbol{d}^{n^T\!,o}$ be, respectively, the vectors of anticipated base generation and demand at all distribution system nodes $\mathcal{N}^{n^T}$, c) $P_{A(n)n}^{n^T}$ and $Q_{A(n)n}^{n^T}$ be, respectively, the real and reactive power flowing over the line connecting $A(n)$ and $n$,
%measured at the sending node ${A}(n)$, 
with a maximum apparent power flow denoted by $S_{A(n)n}^{n^T\!,\textrm{max}}$, %c) $l_{A(i)i}$ be the magnitude squared of the current flowing over $\{A(i),i\}\in\mathcal{L}$, 
and d) $v_n^{n^T}$ be the magnitude squared of the voltage at node $n\in\mathcal{N}^{n^T}$, with upper and lower limits specified, respectively, by $v_n^{n^T\!,\textrm{max}}$ and $v_n^{n^T\!,\textrm{min}}$. 
Here, $P_{A(n_0^{n^T})n_0^{n^T}}^{n^T}$ and $Q_{A(n_0^{n^T})n_0^{n^T}}^{n^T}$ denote the interface flows with the transmission grid, which are equal, respectively, to $T_{n^T}^p$ and $T_{n^T}^q$ defined on the transmission side.
In addition, for the line parameters, we let $r_{A(n)n}^{n^T}$ and $x_{A(n)n}^{n^T}$ be, respectively, the resistance and reactance of line $\{A(n),n\}\in\mathcal{L}^{n^T}$.

We next introduce the formulation of disjoint transmission and distribution markets, followed by the common market model, which joins the distribution and transmission-level markets. 

\subsection{Disjoint Transmission-Level Market}

We consider that the TSO's anticipated base schedule $\boldsymbol{p}^{T,o}$ and $\boldsymbol{d}^{T,o}$ shows imbalance and/or line congestion which the TSO aims to solve using resources available only from the transmission level (hence, resulting in a disjoint market). These resources are represented by offers to the market from assets connected to the transmission network. %  
Let $\Delta p_n^{T+}$ and $\Delta p_n^{T-}$ correspond, respectively, to the volumes of increase and reduction in generation (respectively, upward and downward flexibility) connected at the transmission bus $n$. In addition, let $\Delta d_n^{T+}$ and $\Delta d_n^{T-}$ correspond, respectively, to the volumes of reduction and increase in demand (i.e., upward and downward flexibility) at node $n$ of the transmission network. In addition, let $c_{p_n}^{T+}$, $c_{p_n}^{T-}$, $c_{d_n}^{T+}$, and $c_{d_n}^{T-}$ represent the unit price offered by adjustable generation and loads for, respectively, $\Delta p_n^{T+}$, $\Delta p_n^{T-}$, $\Delta d_n^{T+}$, and $\Delta d_n^{T-}$. %As congestion management is taken into consideration, locational information for every bid is needed. Hence, we consider nodal bids. In other words, offers by each flexibility provider is associated with a certain node (i.e. modification in injection or offtake at a certain node). 
% 
%\textcolor{blue}{(AS: We can relax this consideration by considering the set of all lines instead. This could also be done for the distorbtion-level market for flow limits and voltage limits)} 
%In addition, let $\mathcal{C^{T,R+}}$ and $\mathcal{C^{T,R-}}$ be the set of lines whose flow is anticipated to potentially surpass their limit and whose flow is, respectively, in and opposite to their reference directions. The flow over these lines, hence, has to be explicitly constrained in the market clearing formulation to prevent or alleviate congestions. In general, for increased security of operations, all lines could be considered for congestion management, in which case, $\mathcal{C^{T,R+}} = \mathcal{C^{T,R-}}=\mathcal{L}^T$. 
%
The goal of the TSO is to resolve the balancing and congestion issues at the minimum possible cost. Hence, the TSO's problem can be described as follows:

{\small \begin{align}
\min_{\Delta \boldsymbol{p}^{T}\!,\Delta \boldsymbol{d}^{T}} \!\!\!\sum\limits_{n\in\mathcal{N}^T}\!\!\!\big(c_{p_n}^{T+}\!\Delta p_n^{T+}\!\!-\! c_{p_n}^{T-}\!\Delta p_n^{T-}\!\!+\! c_{d_n}^{T+}\!\Delta d_n^{T+}\!\!-\!c_{d_n}^{T-}\!\Delta d_n^{T-}\big),\vspace{-0.2cm} \label{eq:ObjectiveBalT_OnlyT_v1}
\end{align}\vspace{-0.1cm}
Subject to:
\begin{align}\label{eq:NetInjection_OnlyT}
p_n^T\!\!=\!p^{T,o}_n\!\!+\!\Delta p_n^{T+}\!\!\!-\!\Delta p_n^{T-}\!\!\!-\!d^{T,o}_n\!+\! \Delta d_n^{T+}\!\!\!-\!\Delta d_n^{T-}\!, \forall n\in\mathcal{N}^T,
\end{align}\vspace{-0.6cm}
\begin{align}\label{eq:PowerFlowBalT_OnlyT_v1}
P_{ij}^T=\sum\limits_{n=1}^{N^T}p_n^T X_{(i,j),n} -\sum\limits_{n\in \mathcal{N}^D}T^p_n X_{(i,j),n},\,\, \forall \{i,j\}\in\mathcal{L}^T,
\end{align}
\vspace{-0.4cm}
\begin{align}\label{eq:NodeBalanceBalT_OnlyT_v1}
p_n^T-\sum\limits_{j s.t. \{i,j\}\in\mathcal{L}^T}P_{ij}^T=0: \,\,\,(\lambda^T_n),\,\,\, \forall n\in\mathcal{N}^T\setminus \mathcal{N}^D,
\end{align}\vspace{-0.4cm}
\begin{align}\label{eq:NodeBalanceBalT_OnlyT_v1_atDSONodes}
p_n^T - T^p_n-\sum\limits_{j s.t. \{i,j\}\in\mathcal{L}^T}P_{ij}^T=0: \,\,\,(\lambda^T_n),\,\,\, \forall n\in\mathcal{N}^D,
\end{align}
\vspace{-0.4cm}
\begin{align} \label{eq:CongPlusBalT_OnlyT_v1}
-F_{ij}^{T,\textrm{max}}\leq P_{ij}^T\leq F_{ij}^{T,\textrm{max}}, \forall \{i,j\}\in\mathcal{L}^T,
\end{align}
\vspace{-0.6cm}
%\begin{align} \label{eq:CongMinBalT_OnlyT_v1}
%-P_{ij}\leqslant F_{ij}^{\textrm{max}}, \forall \{i,j\}\in %\mathcal{C^{T,R-}},
%\end{align}
%\vspace{-0.2cm}
\begin{align}\label{eq:TGenBidLimBalT_OnlyTp_v1}
0\leq \Delta p_n^{T+} \leq \Delta p_n^{T+\textrm{,max}}, 
%\, \forall i\in\mathcal{N}^T, 
%\end{align}
%\vspace{-0.2cm}
%\begin{align}\label{eq:TGenBidLimBalT_OnlyTm_v1}
0\leq \Delta p_n^{T-} \leq \Delta p_n^{T-\textrm{,max}}, \, \forall i\in\mathcal{N}^T, 
\end{align}
\vspace{-0.6cm}
\begin{align} \label{eq:TDemBidLimBalT_OnlyTp_v1}
0\leq \Delta d_i^{T+} \leq \Delta d_i^{T+\textrm{,max}}, % 
%\, \forall i\in\mathcal{N}^T,
%\end{align}
%\vspace{-0.2cm}
%\begin{align} \label{eq:TDemBidLimBalT_OnlyTm_v1}
0\leq \Delta d_i^{T-} \leq \Delta d_i^{T-\textrm{,max}},  \, \forall i\in\mathcal{N}^T.
\end{align}}
\vspace{-0.4cm}
%\begin{align} \label{eq:BalanceBalT_OnlyT_v1_Fulll} 
%\sum\limits_{i=1}^{N^T}\left(P^o_i+\Delta P_i^{T+} -\Delta P_i^{T-} -D^o_i + \Delta D_i^{T+} - \Delta D_i^{T-}\right)-\sum\limits_{i \in \mathcal{N}^D}T^p_i=0:\,\,\, (\lambda_B^T), [\textrm{not explicitly needed}].
%\end{align}

Here, $\Delta \boldsymbol{p}^{T}$ and $\Delta \boldsymbol{d}^{T}$ are the vectors grouping, respectively, $\Delta p_n^{T+}$ and $\Delta p_n^{T-}$, and $\Delta d_n^{T+}$ and $\Delta d_n^{T-}$, at all nodes $n\in\mathcal{N}^T$. (\ref{eq:NetInjection_OnlyT}) defines the net injection at node $n\in\mathcal{N}^T$. (\ref{eq:PowerFlowBalT_OnlyT_v1}) consists of the power flow equations
%, (where $P_{ij}$ is the flow over transmission line $\{i,j\}\in\mathcal{L}^T$) 
over all the transmission lines expressed using the generation shift factors, $X_{(i,j),n}$.
%We let $\boldsymbol{\Xi}$ denote the matrix of shift factors.
%Here, the sensitivity factor $\chi_{(i,j),n}$ captures the change in the flow over line $\{i,j\}\in\mathcal{L}^T$ due to a change in injection or offtake at node $n\in\mathcal{N}^T$.
(\ref{eq:NodeBalanceBalT_OnlyT_v1}) is the energy balance equation at node $n\in\mathcal{N}^T\setminus \mathcal{N}^D$, while  (\ref{eq:NodeBalanceBalT_OnlyT_v1_atDSONodes}) is the energy balance equation at interface node $n\in\mathcal{N}^D$. %, where $T^p_i$ is the active power transfer to the distribution system connected at node $i\in\mathcal{N}_D$. 
$\lambda_n^T$ denotes the Lagrange multipliers of the nodal energy balance constraints at each node $n\in\mathcal{N}^T$.
In this disjoint transmission system market formulation, $T^p_n$ is considered to be a constant (i.e. not dependent on any decision variable) to reflect the case that balancing and congestion management in this case are to be resolved using flexibility available only at the transmission system. 
% can be deleted ---- 
%\begin{comment}
%This formulation explicitly uses the power injections into the distribution systems, rather than implicitly accounting for these injections as part of the loads $\boldsymbol{d}^{T,o}$ for all $n\in\mathcal{N}^D$. Hence, the loads represented by $d^o_n$ (and which can be adjusted upward or downward based on the offered flexibility $\Delta d_n^{T+}$ and $\Delta d_n^{T-}$) at all the nodes do not cover the power withdrawn (or injected from, if negative) the distribution systems but rather correspond to large industrial loads.
%\end{comment}
%-----
In addition, (\ref{eq:CongPlusBalT_OnlyT_v1}) represents the congestion prevention constraints, while 
(\ref{eq:TGenBidLimBalT_OnlyTp_v1}) and  (\ref{eq:TDemBidLimBalT_OnlyTp_v1}) capture the bid limits. 
%represent the limits of the submitted bids by FSP $i$. 
%Constrain~(\ref{eq:BalanceBalT_OnlyT_v1_Fulll}) consists of the energy balance constraint over the entire transmission system and is not explicitly needed but could be added for the computation of its Lagrange multiplier ($\lambda_B^T$).

%Pricing after market clearing could be done following a nodal pricing approach in which case up regulation activation at a node $i$ is remunerated at a price equal to $\lambda_i^T$. Similarly, the cost payed by a down-regulating resource at node $i$ is priced based on $\lambda^T_i$. 

%If uniform pricing, rather than nodal pricing, is to be followed, then the balancing price reflecting the cost of the last cleared unit could be used for the pricing mechanism. In that case, $\lambda_{B}^T $ can be used as the uniform market price. If a pay-as-bid pricing mechanisms is to be followed, then the accepted bids will be priced based on their submitted offers (i.e., $c_{ P_i}^{T+}$, $c_{P_i}^{T-}$, $c_{ D_i}^{T+}$, $c_{D_i}^{T-}$).

\subsection{Disjoint Distribution-Level Market}\label{subsec:DisjointDSO}

The disjoint market for congestion management at the DSO--$n^T$ level uses solely distribution grid flexibility %(i.e. modification in injections and demand in the distribution grid) 
to solve congestion issues resulting from the anticipated base generation and load profiles, $\boldsymbol{p}^{n^T\!,o}$ and $\boldsymbol{d}^{n^T\!,o}$, at all distribution system nodes $\mathcal{N}^{n^T}$.  
Let $\Delta p_n^{n^T+}$ and $\Delta p_n^{n^T-}$ correspond, respectively, to the volumes of increase and reduction in generation (corresponding to, respectively, upward and downward flexibility) connected at the DSO--$n^T$ node $n$. In addition, we let $\Delta d_n^{n^T+}$ and $\Delta d_n^{n^T-}$ represent upward and downward flexibility at node $n$ (i.e., the volumes of reduction and increase in demand at that node). In addition, similarly to the cost structure on the transmission level, we  let $c_{p_n}^{n^T+}$, $c_{p_n}^{n^T-}$, $c_{d_n}^{n^T+}$, and $c_{d_n}^{n^T-}$ represent the unit prices offered by adjustable generation and loads on the distribution level for, respectively, $\Delta p_n^{n^T+}$, $\Delta p_n^{n^T-}$, $\Delta d_n^{n^T+}$, and $\Delta d_n^{n^T-}$. % 
%As can be seen from the bid definition, locational information for every bid is needed, as congestion management requires locational information. Thus, the submitted bids by the FSPs on the distribution level are nodal bids where the index $i$ ranges in $\mathcal{N}_{n^T}\setminus 0_{n^T}$. 

%For line flow congestion management, we let $\mathcal{C}_{n^T}$ be the set of lines whose flow is anticipated to potentially surpass their limit. Such lines are then considered to be critical and their flows must be explicitly constrained in the market clearing formulation. Such critical lines can readily be extended and defined, for example, to include all the set of distribution lines. In addition, we let $\mathcal{V}^{\textrm{max}}_{n^T}$ be the set of nodes whose voltage magnitudes are expected to be at risk of exceeding their maximum limits, and $\mathcal{V}^{\textrm{min}}_{n^T}$ be the set of nodes whose voltage magnitudes are expected to be below their minimum limits. These are the nodes whose voltage magnitudes are to be regulated using the congestion management market. In this regard, upper and lower voltage limits can be imposed on all the nodes of the distribution system in which case $\mathcal{V}^{\textrm{max}}_{n^T}=\mathcal{V}^{\textrm{min}}_{n^T}=\mathcal{N}_{n^T}$. 

%We next consider an economic dispatch formulation for this congestion management market which is based on the LinDistFlow power flow model presented in \cite{deliverable}.  

As the DSO--$n^T$'s goal is to resolve congestion issues at minimum cost, its problem can be formulated as follows.
{\small \begin{flalign}
\min_{\Delta \boldsymbol{p}^{n^T},\Delta \boldsymbol{d}^{n^T},\boldsymbol{q}^{n^T}}\sum\limits_{n\in N^{n^T}}\big(&c_{p_n}^{n^T+}\Delta p_n^{n^T+} -c_{p_n}^{n^T-}\Delta p_n^{n^T-}&\nonumber\\\vspace{-0.4cm}
&+c_{d_n}^{n^T+}\Delta d_n^{n^T+}-c_{d_n}^{n^T-}\Delta d_n^{n^T-}\big),&\vspace{-0.4cm} \label{eq:ObjectiveCM_onlyD}
\end{flalign}\vspace{-0.1cm}
Subject to:
\begin{align}
p_n^{n^T}\!\!\!\!=\!p^{n^T\!\!,o}_n\!\!+\!\Delta p_n^{n^T\!+}\!\!\!-\Delta p_n^{n^T\!-}\!\!\!-d^{n^T\!\!,o}_n\!\!+\! \Delta d_n^{n^T\!+}\!\!\!-\!\Delta d_n^{n^T\!-}\!, \forall n\!\in\! \mathcal{N}^{n^T}\!\!,\label{eq:AdjustedInjection_OnlyD_Lin1} 
\end{align}\vspace{-0.4cm}
\begin{align}
p_n^{n^T}\!+\!P_{A(n)n}^{n^T}\!-\!\!\!\sum\limits_{k\in\mathcal{K}(n)}P_{nk}^{n^T}\!\!=0,\,\,\forall n\in\mathcal{N}^{n^T}\setminus n_0^{n^T}:(\lambda_n^{n^T}),\,\,\label{eq:LinDistFlow1_OnlyD_Lin1}\\
q_n^{n^T}+Q_{A(n)n}^{n^T}-\sum\limits_{k\in\mathcal{K}(n)}Q_{nk}^{n^T}=0\,\, \forall n\in\mathcal{N}^{n^T}\setminus n_0^{n^T},\label{eq:LinDistFlow2_OnlyD_Lin}
\end{align}\vspace{-0.4cm}
\begin{align}
T^p_{n^T}-\sum\limits_{k\in\mathcal{K}(n_0^{n^T})}P_{n_0^{n^T}k}^{n^T}=0,\,\,:(\lambda_{n_0^{n^T}}), \label{eq:LinDistFlow1_root_OnlyD_Lin1}\\
T^q_{n^T}-\sum\limits_{k\in\mathcal{K}(n_0^{n^T})}Q_{n_0^{n^T}k}^{n^T}=0,\,\label{eq:LinDistFlow2_root_OnlyD_Lin1}
\end{align}\vspace{-0.3cm}
\begin{align}
v_n^{n^T}\!\!=\!v_{\mathcal{A}(n)}^{n^T}\!-\!2r_{A(n)n}^{n^T}P_{A(n)n}^{n^T}\!-\!2x_{A(n)n}^{n^T}Q_{A(n)n}^{n^T},\,\forall n\in\mathcal{N}^{n^T}\!\!\setminus n_0^{n^T},\label{eq:LinDistFlow3_DSOLocal_OnlyD_Lin1}
\end{align}\vspace{-0.3cm}
\begin{align}
\alpha_m P_{A(n)n}^{n^T}\!\!+\!\beta_m Q_{A(n)n}^{n^T}\!\!+\!\delta_m  S_{A(n)n}^{n^T,\textrm{max}}\!\leq 0, \forall m \!\in\!\mathcal{M}, \{A(n),n\}\! \in\!\mathcal{L}^{n^T}\!\!\!, \label{eq:FlowLimit_OnlyD_Lin1}
\end{align}\vspace{-0.4cm}
\begin{align} \label{eq:VMaxCM_onlyD_Lin1}
v_n^{n^T,\textrm{min}}\leq v^{n_T}_n\leqslant v_n^{n^T,\textrm{max}}, 
\forall n\in\mathcal{N}^{n^T},
\end{align}
\vspace{-0.4cm}
\begin{align}\label{eq:ReactivePowerLim_OnlyD_Lin1}
q_n^{n^T,\textrm{min}} \leq q^{n^T}_n \leq q_n^{n^T,\textrm{max}},\, \forall n\in\mathcal{N}^{n^T}, 
\end{align}\vspace{-0.4cm}
\begin{align}\label{eq:LimExchangeQ_LocalMrkt2}
T_{n^T}^{q,\textrm{min}}\leq T_{n^T}^q\leq T_{n^T}^{q,\textrm{max}}, 
\end{align}\vspace{-0.4cm}
\begin{align}\label{eq:GenBidLimCMp_onlyD1}
0\leqslant \!\Delta p_n^{n^T+} \!\leqslant \Delta p_n^{n^T\!+,\textrm{max}}, 
%\,\, \forall i\in\mathcal{N}_{n^T}\setminus 0_{n^T}, 
%\end{align}
%\vspace{-0.2cm}
%\begin{align}\label{eq:GenBidLimCMm_onlyD1}
0\leqslant \Delta p_n^{n^T\!-}\! \leqslant \Delta p_n^{n^T\!-,\textrm{max}}, \,\forall n\in\mathcal{N}^{n^T}, 
\end{align}
\vspace{-0.4cm}
\begin{align} \label{eq:DimBidLimCMp_onlyD1}
0\leqslant \!\Delta d_n^{n^T+} \leqslant \Delta d_n^{n^T\!+,\textrm{max}}\!\!, 
%\,\, \forall i\in\mathcal{N}_{n^T}\setminus 0_{n^T},
%\end{align}
%\vspace{-0.2cm}
%\begin{align} \label{eq:DimBidLimCMm_onlyD1}
0\leqslant \!\Delta d_n^{n^T\!-} \!\leqslant \Delta d_n^{n^T\!-,\textrm{max}},\,\forall n\in\mathcal{N}^{n^T}\!\!.
\end{align}}
\vspace{-0.4cm}
%\begin{align} \label{eq:BalanceCM_onlyD1}
%\sum\limits_{i=1}^{N_{n^T}}\left(\Delta p_i^{n_T+} - \Delta p_i^{n_T-} +\Delta d_i^{n_T+} - %\Delta d_i^{n_T-}\right)=0:\,\,\, (\lambda_B^{n^T}) \,\,\textrm{[Not explicitly needed]}. 
%\end{align}

Here, $\Delta \boldsymbol{p}^{n^T}$ and $\Delta \boldsymbol{d}^{n^T}$ are vectors grouping, respectively, $\Delta p_n^{n^T+}$ and $\Delta p_n^{n^T-}$, and $\Delta d_n^{n^T+}$ and $\Delta d_n^{n^T-}$ at each $n\in\mathcal{N}^{n^T}$. (\ref{eq:AdjustedInjection_OnlyD_Lin1}) returns the net generation, $p_n^{n^T}$, which depends on the base generation and load profiles and the activated flexibility at node $n$.  
The equality constraints in (\ref{eq:LinDistFlow1_OnlyD_Lin1})-(\ref{eq:LinDistFlow3_DSOLocal_OnlyD_Lin1}) follow directly from the \emph{LinDistFlow} formulation \cite{LinDistFlow}, representing the linearized power flow equations in radial distribution networks. 
%The \emph{LinDistFlow} formulation approximates the common relaxed branch flow equations~\cite{LinDistFlow} (known as \emph{DistFlow}) by neglecting the branch loss and shunt components.   
%As seen from (\ref{eq:LinDistFlow1_root_OnlyD_Lin1}) and (\ref{eq:LinDistFlow2_root_OnlyD_Lin1}), and without loss of generality, the root node, $n_0^{n^T}$, considers solely the power injection from the transmission side, i.e. the interface power flow between the transmission system and the distribution system, $T^p_{n^T}$ and $T^q_{n^T}$, respectively. 
In a local market mechanism (i.e. disjoint DSO--$n^T$ level market), the interface flow $T^p_{n^T}$ is considered to be constant, i.e., not a decision variable nor a dependent variable %that is a function of decision variables 
to be optimized. Hence, this reflects that congestion management on the distribution level is to be resolved solely using flexibility provided from within the distribution system. 
We denote the Lagrange multiplier of constraint (\ref{eq:LinDistFlow1_OnlyD_Lin1}) and (\ref{eq:LinDistFlow1_root_OnlyD_Lin1}) by $\lambda_n^{n^T}$.
Constraint~(\ref{eq:FlowLimit_OnlyD_Lin1}) is a linearization of the complex flow limit constraint. This linearization, as proposed in~\cite{DGHostingCapacity}, is a polygonal inner-approximation that transforms the feasibility region of the flow limit constrain from a circle of radius $S_{A(n)n}^{n^T,\textrm{max}}$ into a polygon whose number of edges are given by the size of the approximation set $\mathcal{M}$. The values of $\alpha_m$, $\beta_m$, and $\gamma_m$ define this polygon such that all its vertices would lie on the original feasibility circle of radius $S_{A(n)n}^{n^T,\textrm{max}}$ (a detailed explanation of this approximation is presented in~\cite{DGHostingCapacity}).
%\footnote{This linear distribution-level market model based on LinDistFlow power flow equations and inner polygonal approximation of line flow constraints, has been presented in~\cite{CIRED21}. \cite{CIRED21} shows that this formulation closely reflect the outcomes of its second-order cone programming (SOCP) counterpart, while significantly reducing its complexity and computational burden, which increases its applicability in practice.}
Constraints~(\ref{eq:VMaxCM_onlyD_Lin1}) and (\ref{eq:ReactivePowerLim_OnlyD_Lin1}) capture the limits on the nodal voltage magnitudes and reactive power injections, to ensure operational stability and the real-reactive power operational and capacity limits of load and generation, while (\ref{eq:LimExchangeQ_LocalMrkt2}) enforces a limit on the reactive power transfer with the transmission grid, $T^q_{n^T}$.  
Constraints (\ref{eq:GenBidLimCMp_onlyD1}) and~(\ref{eq:DimBidLimCMp_onlyD1}) reflect the limits of the submitted bids.

%Pricing after market clearing could be done based on a nodal approach, where the nodal prices are given by $\lambda_i^{n^T}$. If a uniform pricing is to be followed, the uniform market price would then correspond to $\lambda_B^{n^T}$, while in a pay-as-bid mechanism, the renumeration scheme is based on the submitted offer prices (i.e., $c_{p_i}^{n^T+}$, $c_{p_i}^{n^T-}$, $c_{d_i}^{n^T+}$, $c_{d_i}^{n^T-}$).

\subsection{Common Market Model}
In the common market, the TSO can readily use resources offered from within the different distribution systems connected to its transmission network as well as resources connected to its transmission network to perform balancing and congestion management. Concurrently, the DSOs can use resources offered from their distribution networks for congestion management. 
%This formulation includes the operational constraints of all the systems to prevent any activation of resources from leading to constraint violation of any of the grids involved. % 
%
%Hence, 
In the common market, flexibility resources (i.e. submitted bids) are accessible to all participating SOs and the market is jointly cleared, in a collaborative manner, to optimally meet the needs of all the SOs while abiding by the operational limits (i.e. constraints) of all the grids involved. 
The proposed common market formulation, incorporating the TSO and the $\mathcal{N}^D$ DSOs, combines the disjoint market models as follows:\vspace{-0.2cm}

{\small \begin{flalign}
&\min_{\Delta \boldsymbol{p},\Delta \boldsymbol{d}, \boldsymbol{q}}\!\Big[\!\!\sum\limits_{n\in\mathcal{N}^T}\!\!\!\left(c_{p_n}^{T+\!}\!\Delta p_n^{T+\!}\!\!-\!c_{p_n}^{T-\!}\!\Delta p_n^{T-\!}\!\!+\!c_{d_n}^{T+\!}\!\Delta d_n^{T+\!}\!\!-\! c_{d_n}^{T-\!}\Delta d_n^{T-\!}\right)&\nonumber\\
&\!+\!\!\!\!\!\sum\limits_{n^T\!\in\mathcal{N}^D}\!\!\sum\limits_{n\in\mathcal{N}^{n^T}}\!\!\!\!\left(\!c_{p_n}^{n^T\!\!+}\!\Delta p_n^{n^T\!\!+}\!\!\!-\!c_{p_n}^{n^T\!\!-}\!\Delta p_n^{n^T\!\!-}\!\!\!+\!c_{d_n}^{n^T\!\!+}\!\Delta d_n^{n^T\!\!+}\!\!\!-\!c_{d_n}^{n^T\!\!-}\!\Delta d_n^{n^T\!-}\!\right)\!\Big],& \label{eq:ObjectiveCMandBal_CommonMrkt}
\end{flalign}%\vspace{-0.2cm}
Subject to:\vspace{-0.4cm}
\begin{align}
\textrm{(\ref{eq:NetInjection_OnlyT})--(\ref{eq:TDemBidLimBalT_OnlyTp_v1}), (\ref{eq:AdjustedInjection_OnlyD_Lin1})--(\ref{eq:DimBidLimCMp_onlyD1})\,} \forall n^T\in\mathcal{N}^D, 
\end{align}
\begin{align}\label{eq:LimsumPDSOBalT_CommonMrkt}
\textrm{and\,\,\,\,\,}T_{n^T}^{p, \textrm{min}}\leq T_{n^T}^p\leq T_{n^T}^{p, \textrm{max}} \,\, \forall n^T\in\mathcal{N}^D. 
\end{align}}\vspace{-0.4cm}
%\begin{align}\label{eq:LimsumPDSOBalT_CommonMrkt_reactive}
%T_i^{q,\textrm{min}}\leq T_i^q\leq T_i^{q, \textrm{max}}, \,\, \forall %i\in\mathcal{N}^D. 
%\end{align}    
%\vspace{-0.4cm}

$\Delta \boldsymbol{p}$ and $\Delta \boldsymbol{d}$ are the vectors of generation and demand flexibility variables (i.e., respectively, $\Delta\boldsymbol{p}^T$ and $\Delta\boldsymbol{p}^{n^T}$ for all $n^T\in\mathcal{N}^D$, and $\Delta\boldsymbol{d}^T$ and $\Delta\boldsymbol{d}^{n^T}$ for all $n^T\in\mathcal{N}^D$). In addition, $\boldsymbol{q}$ is the vector of $\boldsymbol{q}^{n^T}$ for all $n^T\in\mathcal{N}^D$. 
%The common market model combines the disjoint markets in one formulation including all the submitted flexibility bids in all of the systems and the constraints of all of the systems. 
The common market formulation sums the objective functions of the disjoint markets and uses all the operational and bid level constraints from these markets. 
An additional constraint in the common market is (\ref{eq:LimsumPDSOBalT_CommonMrkt}). 
%and (\ref{eq:LimsumPDSOBalT_CommonMrkt_reactive}). 
In this regard, in contrast to the disjoint markets, 
%which considers the power transfer between the transmission and distribution systems as constant, 
the common market model treats $T^p_{n^T}$ as a dependent  variable (rather than a constant as in the disjoint markets) with limits shown in (\ref{eq:LimsumPDSOBalT_CommonMrkt}). %
%and (\ref{eq:LimsumPDSOBalT_CommonMrkt_reactive}). 
This enables the interaction between the previously disjoint markets to jointly procure the flexibility needed by all SOs.  
%and (\ref{eq:LimsumPDSOBalT_CommonMrkt_reactive}). 
\begin{remark}\label{Re:Sub-common} 
We note that the common market formulation in (\ref{eq:ObjectiveCMandBal_CommonMrkt})--(\ref{eq:LimsumPDSOBalT_CommonMrkt}) can be readily adapted to reflect a sub-common market setting joining any subset of the SOs. This is readily achieved by replacing $\mathcal{N}^D$ by any defined subset of DSOs in (\ref{eq:ObjectiveCMandBal_CommonMrkt})--(\ref{eq:LimsumPDSOBalT_CommonMrkt}), and by optionally including or excluding the TSO-level market from the formulation.
\end{remark}

\kern-1.35em % blank space reduction

\subsection{Compact Formulation}

Due to it's linearity, and for ease of notation, the common market model can be readily expressed using a compact linear programming formulation. This is carried out next, which will be useful in defining the cooperative game framework.  
  
We introduce $\bd{x}_0$ as the column vector that contains the TSO's decision variables: $\bd{x}_0 \triangleq \textrm{col}\begin{pmatrix}
    \bd{\Delta p}^{T+}, \bd{\Delta p}^{T-}, \bd{\Delta d}^{T+}, \bd{\Delta d}^{T-}
\end{pmatrix}$. 
The TSO's feasibility set $\mc{X}_0$ is defined based on the constraints that depend only on the TSO's decision variables. %At the distribution level, each DSO--$n^T$ is connected to a transmission node $n^T \in \mc{N}^T$. 
Each DSO--$n^T$'s decision variable column vector, denoted by $\bd{x}_{n^T}$, contains 
\begin{comment}
real, reactive power supply and demand side, and net reactive power withdrawal:
\end{comment}
$\bd{x}_{n^T} \triangleq \textrm{col}\begin{pmatrix}
    \bd{\Delta p}^{n^T +},\bd{\Delta p}^{n^T -}, \bd{\Delta d}^{n^T +}, \bd{\Delta d}^{n^T -}, \bd{q}^{n^T}
    \end{pmatrix}$.
We define $\mc{X}_{n^T}$ as the feasibility set of DSO--$n^T$. The set is made of constraints that depend only on the DSO's decision variables. 
We let $\bd{x} \in \mc{X}_0 \times \prod_{n^T}\mc{X}_{n^T}$ be the vector that contains the joint decisions of the TSO and DSOs.
Dependent variables, which can be expressed as functions of the TSO and DSOs' decision variables, are concatenated in a column vector $\bd{z}$. Precisely, we set: $\bd{z} \triangleq \textrm{col}\Big(\bd{p}^T, \bd{P}^T, (\bd{p}^{n^T})_{n^T \in \mc{N}^D}, (\bd{P}^{n^T})_{n^T\in\mc{N}^D}, (\bd{Q}^{n^T})_{n^T\in\mc{N}^D},$ 
$(\bd{v}^{n^T})_{n^T\in\mc{N}^D},\bd{T}^p, \bd{T}^q\Big)$. Constants such as $\bd{p}^{T,o}$, $\bd{d}^{T,o}$, $(\bd{p}^{n^T\!,o})_{n^T\in\mc{N}^D}$, $(\bd{d}^{n^T\!,o})_{n^T\in\mc{N}^D}$, $\bd{\Delta p}^{T+,\max}$,$\bd{\Delta p}^{T-,\max}$, $\bd{\Delta d}^{T+,\max}$, $\bd{\Delta d}^{T-,\max}$, ($\bd{\Delta p}^{n^T+,\max}$,$\bd{\Delta p}^{n^T-,\max}$, $\bd{\Delta d}^{n^T+,\max}$, $\bd{\Delta d}^{n^T-,\max})_{n^T\in\mc{N}^D}$, $(\bd{v}^{n^T,\max})_{n^T\in\mc{N}^D}$, $(\bd{v}^{n^T,\min})_{n^T\in\mc{N}^D}$, $(\bd{q}^{n^T,\max})_{n^T\in\mc{N}^D}$, $(\bd{q}^{n^T,\min})_{n^T\in\mc{N}^D}$, $\bd{T}^{q,\max}$, $\bd{T}^{q,\min}$, $\bd{T}^{p,\max}$, $\bd{T}^{p,\min}$, $\bd{\Xi}\big((i,j),n\big)=X_{(i,j),n}$, $(\boldsymbol{r}^{n^T}$, $\boldsymbol{x}^{n^T})_{n^T\in\mathcal{N}^D}$, $\boldsymbol{\alpha}, \boldsymbol{\beta}, \boldsymbol{\delta}$, $\boldsymbol{F}^{T\!,\textrm{max}}$, and $(\boldsymbol{S}^{n^T\!,\textrm{max}})_{n^T\in\mc{N}^D}$, are taken as input parameters. 
%
%\textcolor{blue}{AS: $\bd{T}^p$ is a dependent variable based on (64), (66) and (85) in the mathematical formulation document.}
%
We also define the TSO's cost vector 
$\bd{c}_0 \triangleq \begin{pmatrix}
\bd{c}_p^{T+}, -\bd{c}_p^{T-}, \bd{c}_d^{T+}, -\bd{c}_d^{T-}
\end{pmatrix}$
and DSO--$n^T$'s cost vector 
$\bd{c}_{n^T} \triangleq \begin{pmatrix}
\bd{c}_p^{n^T+}, -\bd{c}_p^{n^T-}, \bd{c}_d^{n^T+}, -\bd{c}_d^{n^T -}, \bd{0}
\end{pmatrix}.$

%In the common market model formulation, the TSO minimizes the flexibility activation cost on his network. 
Based on this compact notation, the TSO's objective function can be written as $\Phi_0(\bd{x}_0) \triangleq \bd{c}_0^t \bd{x}_0$, while the objective function of each DSO-$n^T$ is given by $\Phi_n(\bd{x}_{n^T}) \triangleq \bd{c}_{n^T}^t \bd{x}_{n^T}$, where $.^t$ is the transpose operator.
%; this notation will be used throughout the paper. 
%
%
The social cost of the common market, introduced in~\eqref{eq:ObjectiveCMandBal_CommonMrkt}, is defined as the sum of the TSO and DSOs' objective functions, i.e., $\Phi(\bd{x}) \triangleq \Phi_0(\bd{x}_0)+\sum_{n^T \in \mc{N}^T} \Phi_{n^T}(\bd{x}_{n^T})$.
To write the common market model \eqref{eq:ObjectiveCMandBal_CommonMrkt}-\eqref{eq:LimsumPDSOBalT_CommonMrkt} in a compact matrix form, a change of notations is needed. For that purpose, we order the DSOs in ascending order form $1$ to $N=|\mathcal{N}^D|$ based on the values of their initial $n^T$. In addition, we denote $\bd{x} \triangleq \Big( \bd{x}_0, (\bd{x}_n)_n\Big)$ as the concatenation of the TSO and DSOs' decision variables.

As the common market model in \eqref{eq:ObjectiveCMandBal_CommonMrkt}-\eqref{eq:LimsumPDSOBalT_CommonMrkt} is linear, it can readily be presented as a standard compact linear program (LP), using our defined vector notation, as follows:\vspace{-0.4cm}

{\small \begin{subequations} \label{eq:compact_model}
\begin{align}
\textrm{(LP)} \hspace{0.8cm}    \min_{\bd{x},\bd{z}} & \hspace{0.2cm} \Phi(\bd{x}), \label{eq:compact1}\\
    s.t. & \hspace{0.2cm} \bd{A} \bd{x} + \bd{B} \bd{z} \leq \bd{d}, & \hspace{0.1cm} (\bd{\lambda}) \label{eq:compact2} \\ 
         &\hspace{0.2cm}  \bd{x}_0 \in \mc{X}_0, \label{eq:compact3} \\
         & \hspace{0.2cm} \bd{x}_{n} \in \mc{X}_{n}, \forall n \in \mc{N},  \label{eq:compact4} \\
         & \hspace{0.2cm} \bd{z} \in \mc{Z}. \label{eq:compact5}
\end{align}\vspace{-0.4cm}
\end{subequations}}
%The explicit expressions of $\bd{A}$, $\bd{B}$, and $\bd{d}$ can be found in Appendix~\ref{appendix}. 

The steps towards this compact formulation are presented in the Appendix. Note that \eqref{eq:compact2} can be equivalently written as $(\bd{A}_0\bd{x}_0+\bd{B}_0\bd{z}_0)+\sum_n(\bd{A}_n\bd{x}_n+\bd{B}_n\bd{z}_n)\leq \bd{d}$ to differentiate between the TSO and the different DSO $n$'s variables. We let $\bd{\lambda}$ be the dual variable of \eqref{eq:compact2}.

Constraints \eqref{eq:compact3} and \eqref{eq:compact4} can be analytically written as $\Psi_0(\bd{x}_0)\leq \bd{0}$ and $\Psi_n(\bd{x}_n) \leq \bd{0}, \forall n \in \mc{N}$ respectively. We denote $\bd{\mu}_0$ and $(\bd{\mu}_n)$ as the associated dual variables.

The Lagrangian function associated with the optimization problem \eqref{eq:compact_model} is defined as follows:\vspace{-0.4cm} 

\begin{small}\begin{align} \label{eq:lagrangien}
\mc{L}(\bd{x},\bd{z},\bd{\lambda}, \bd{\mu}) \triangleq & \Phi(\bd{x}) + \bd{\lambda}^t \Big( \bd{A} \bd{x} + \bd{B} \bd{z} -  \bd{d} \Big) + \bd{\mu}_0^t \Psi_0(\bd{x}_0) \nonumber \\
+& \sum_n \bd{\mu}_n^t \Psi_n(\bd{x}_n).\vspace{-0.2cm}
\end{align}\end{small}

In addition, we introduce the relaxed formulation of \eqref{eq:compact_model} where the coupling constraint \eqref{eq:compact2} is replaced by a penalty in the objective:\vspace{-0.4cm}

\begin{small}\begin{subequations}\label{eq:relax}
\begin{align} 
    \max_{\bd{\lambda} \geq \bd{0}} \min_{\bd{x} \in \mc{X},\bd{z} \in \mc{Z}} & \hspace{0.2cm} \Big( \Phi(\bd{x})+\bd{\lambda}^t (\bd{A}\bd{x}+\bd{B}\bd{z}-\bd{d}) \Big), \label{eq:relax1} \\
    s.t. & \hspace{0.2cm} \bd{x}_0 \in \mc{X}_0, \label{eq:relax2} \\
    & \hspace{0.2cm} \bd{x}_n \in \mc{X}_n, \forall n \in \mc{N}, \label{eq:relax3}\\
    & \hspace{0.2cm} \bd{z} \in \mc{Z}. \label{eq:relax4}
\end{align}\vspace{-0.4cm}
\end{subequations}\end{small}
%\begin{lemma} \label{lemma:relax_eq}
%Assuming Slater's condition holds, problems \eqref{eq:compact_model} and \eqref{eq:relax} are equivalent.
%\end{lemma}

%\begin{proof}
The dual function of problem (\ref{eq:compact_model}) is $g(\bd{\lambda},\bd{\mu})\triangleq \min_{\bd{x},\bd{z}} \mc{L}(\bd{x},\bd{z},\bd{\lambda},\bd{\mu})$ where $\mc{L}(.)$ is defined in \eqref{eq:lagrangien}. The associated dual problem is $\max_{\bd{\lambda}\geq \bd{0},\bd{\mu}\geq \bd{0}} g(\bd{\lambda}, \bd{\mu})$. The duality gap is the non-negative number $\Phi(\bd{x}^*)-g(\bd{\lambda}^*,\bd{\mu}^*)$ where $\bd{x}^*$ is solution to the primal problem \eqref{eq:compact_model} and $\bd{\lambda}^*,\bd{\mu}^*$ are solutions to the dual problem. Under constraint qualification (e.g. Slater’s condition) for \eqref{eq:compact_model}, strong duality
implies:\vspace{-0.4cm}

\begin{small}\begin{align}
    \min_{\bd{x}, \bd{z}} \max_{\bd{\lambda}\geq \bd{0}} \mc{L}(\bd{x},\bd{z},\bd{\lambda}) = \max_{\bd{\lambda}\geq\bd{0}} \min_{\bd{x}, \bd{z}} \mc{L}(\bd{x},\bd{z},\bd{\lambda}).
\end{align}\end{small}
%This implies that the solutions to primal \eqref{eq:compact_model} and dual problems coincide. However,

Since $\Psi_0(.), \Psi_n(.), \forall n \in \mc{N}$ and \eqref{eq:compact2}, \eqref{eq:compact5} are all defined through affine functions, Slater's condition is not strictly required and can be replaced with non-strict inequalities in the (LP) feasibility set.

%\end{proof}

\section{Cost Game Definition}\label{sec:CostGame}

We consider a game that is populated by the non-empty set $\mc{N}$ of the $N$ DSOs and the TSO. We will refer to them as the players of the game. A coalition is a subset of $\mc{N}^G \triangleq \mc{N} \cup \{\textrm{TSO}\}$. The grand coalition is the set $\mc{N}^G$ of all SOs (players), with cardinality $N^G$. We cast this game as a characteristic function cost game, $G$, formally defined as follows.

\begin{definition}[Characteristic Function Cost Game \cite{chalkiadakis}] \label{def:characteristic_fct_game}
A characteristic function cost game $G$ is given by a pair $(\mc{N}^G,v)$, where $v:2^{\mc{N}^G} \rightarrow \mathbb{R}$ is a characteristic function, which maps each coalition $C \subseteq \mc{N}^G$ to a real number $v(C)$. The real number $v(C)$ denotes the value of the coalition $C$.
\end{definition}

%\begin{remark}
%The obvious representation of a characteristic function game is to explicitly list every coalition $C \subseteq \mc{N}^G$ together with the associated value $v(C)$. However, it is clear that this representation grows exponentially with the number of players in the game. This means that this representation is limited to a small number of DSOs.
%\end{remark}

The common market in \eqref{eq:compact_model}, and its special cases for a singleton TSO or DSO market (or any sub-common market containing any subset of the SOs as discussed in Remark~\ref{Re:Sub-common}), corresponds to a characteristic function $v$ given by:\vspace{-0.2cm} 

\begin{small}\begin{align}  \label{eq:characteristic_fct}
v(C) = \sum_{n \in C} \bd{c}_n^t \bd{x}_{n}^* \; \textrm{where $C \subseteq \mc{N}^G$},
\end{align}\end{small}\vspace{-0.2cm}

and $(\bd{x}_n^*)_{n \in C}$ is the optimum of the optimization problem\vspace{-0.4cm}

\begin{small}\begin{subequations}\label{eq:coal_pb}
\begin{align}
    \min_{(\bd{x}_n)_{n\in C},(\bd{z}_n)_{n\in C}} & \hspace{1.cm} \sum_{n \in C} \Phi_n(\bd{x}_{n}), \label{eq:subPb}\\
    s.t. & \hspace{1.cm} \bd{A} \bd{x}+\bd{B} \bd{z} \leq \bd{d}, \label{constraint:subPb}\\
    & \hspace{1.cm} \bd{x}_n \in \mc{X}_n, \bd{z}_n \in \mc{Z}_n, \forall n \in C.
\end{align}
\end{subequations}\end{small}\vspace{-0.4cm}

%\begin{remark}
%If the coalitions are reduced to singletons and $\bd{A}_n=\bd{A}, \bd{B}_n=\bd{B}, \forall n \in \mc{N}$, solving \eqref{eq:subPb}-\eqref{constraint:subPb} is equivalent to compute equilibria solutions of a GNEP. It is straightforward to prove that these equilibria are efficient.  
%\end{remark}
As the coalitional value $v(C)$ (i.e., the total cost of the common and sub-common markets) can be divided amongst the members of $C$ in any way that the members of $C$ choose, $G$ is classified as a transferable utility (TU) game~\cite{chalkiadakis}.

%\subsection{Cost Game Basic Properties}
An outcome of the characteristic function game, $G$, consists of two parts: i) a partition of the SOs (i.e. the players) into coalitions, known as the coalition structure; and ii) a cost vector, distributing the value of each coalition among its members. These concepts are formally defined next. %In addition, we make the standard assumption that the value of the empty coalition $\emptyset$ is $0$.

\begin{definition}[Coalition Structure]
Given our game $G=(\mc{N}^G,v)$, a coalition structure over $\mc{N}^G$ is a collection of non-empty subsets $CS=\{C^1,...,C^s\}$ such that $\cup_{j=1}^{s} C^j=\mc{N}^G$ and $C^i\cap C^j=\emptyset, \forall i,j \in \{1,...,s\}$ such that $i \neq j$.
\end{definition}

A vector $\bd{y}=(y_1,...,y_{N^G}) \in \mathbb{R}^{N^G}$ is a cost allocation vector for a coalition structure $CS = \{C^1,...,C^s\}$ over $\mc{N}^G$ if: 
$y_n \geq 0, \forall n \in \mc{N}^G$, $\sum_{n \in C^{s'}} y_n \geq v(C^{s'})$ for any $s' \in \{1,...,s\}$ (feasibility condition). The efficiency of $\bd{y}$ is defined as follows.

\begin{definition}[Efficiency]
A cost allocation vector $\bd{y}$ is efficient if all the coalition cost is distributed amongst coalition members, i.e., $\sum_{n \in C^{s'}} y_n = v(C^{s'}), \forall s' \in \{1,...,s\}$.
\end{definition}

The space of all coalition structures will be denoted $\mc{CS}$. An outcome of $G$ is, hence, a pair $(CS, \bd{y})$. 
%where $CS$ is a coalition structure over $G$ and $\bd{y}$ is a cost allocation vector for $CS$. 
For a cost allocation vector $\bd{y}$, we let $y(C) \triangleq \sum_{n\in C} y_n$ denote the total cost allocation of a coalition $C \subseteq \mc{N}^G$ under $\bd{y}$. By extension, the social cost of the coalition structure $CS$ will be denoted $v(CS) \triangleq \sum_{C \in CS} v(C)$.

%\begin{definition}[Imputation]
%A cost allocation $\bd{y}$ for a coalition structure $CS \in \mc{CS}$ is said to be an imputation if it is efficient and satisfies the Individual Rationality (IR) condition, i.e., $y_n \leq v(\{n\}), \forall n \in \mc{N}^G$. 
%\end{definition}

%We denote $\mc{I}(CS)$ the set of all imputations for a coalition structure $CS$. If a cost allocation vector is an imputation, then each player weakly prefers being in the coalition structure to being on her own. Of course, players might still find it profitable to deviate as a group. This issue will require to introduce core and core-related solution concepts.

For the derivations that ensue, we now recall classical definitions of two subclasses of coalitional games that will be useful thereafter: submodular games and concave games. %Concave cooperative games capture the intuitive property some games have of ``snowballing''. 

\begin{definition}[Submodularity \cite{chalkiadakis}] \label{def:submodularity}
A characteristic function $v$ is said to be submodular if it satisfies\vspace{-0.4cm}

\begin{small}\begin{align*}
    v(C \cup C') + v(C \cap C') \leq v(C) +v(C'),
\end{align*}\end{small}\vspace{-0.4cm}

\noindent for every pair of coalitions $C, C' \subseteq \mc{N}^G$. A game with a submodular characteristic function is said to be concave.
 \end{definition}
 
Concave games have an intuitive characterization in terms of players' marginal contributions: in a concave game, a player (i.e. an SO) is more useful (decreasing the group cost) when it joins a bigger coalition, as formally defined next. 
 
\begin{definition}[Concavity \cite{chalkiadakis}] \label{def:concavity}
A characteristic function game $G$ is concave if and only if for every pair of coalitions $C,C'$ such that $C\subset C'$ and every player $n \in \mc{N}^G \setminus C'$ it holds that
$v(C \cup \{n\}) - v(C) \leq v(C' \cup \{n\}) -v(C')$.
\end{definition}

%\subsection{Cost Game $G$ Core Definition}
To define whether a common market can naturally arise, the stability of the cooperation between the SOs in a common market must be defined and verified.
A stable coalition is a coalition from which no SO has an incentive to deviate. A stable grand coalition, is the coalition including the TSO and all DSOs and which is stable. 
%We introduce a solution concept that captures coalitional stability. 
Consider an outcome $(CS,\bd{y})$ of the cost game $G$. If $y(C) > v(C)$ for some $C \subseteq \mc{N}^G$, the SOs in $C$ could do better by abandoning the coalition structure $CS$ and forming a coalition of their own.
Thus, in this case, the outcome $(CS, \bd{y})$ is unstable. The set of stable outcomes, i.e., outcomes where no subset of SOs (players) has an incentive to deviate, is called the core of $G$.

\begin{definition}\label{eq:TheCore}
The core $\mc{C}(G)$ of the characteristic function game $G=(\mc{N}^G,v)$ is the set of all efficient outcomes $(\mc{CS},\bd{y})$ such that $y(C) \leq v(C), \forall C \subseteq \mc{N}^G$. The core of our cost game $G$ can formally be defined as follows: $\mc{C}(G) \triangleq  \{ \bd{y} \in \mathbb{R}^{N^G} | y(\mc{N}^G)=v(\mc{N}^G) \; \textrm{and} \; y(C) \leq v(C), 
     \hspace{0.2cm} \forall C \subseteq \mc{N}^G \}$.
\end{definition}

%\subsubsection{A naive idea to simulate the core of $G$} 

%To simulate the core of $G$, an idea would be to rely on Shapley's constructive proof of the non-emptiness of the core of convex games. The idea is to fix an arbitrary permutation $p \in \mc{P}_{\mc{N}^G}$, i.e., one-to-one mappings from $\mc{N}^G$ to itself. Given a permutation $p \in  \mc{P}_{\mc{N}^G}$, we denote by $S_p(n)$ the set of all predecessors of $n$ in $p$, i.e., we set $S_p(n) := \{ n' \in \mc{N}^G | p(n') \prec p(n) \}$. The marginal contribution of a player $n$ with respect to a permutation $p$ in a game $G$ is denoted $\Delta_p(n)$ and is given by
%\begin{align} \label{eq:marginal_contribution}
%    \Delta_p(n) \triangleq v(S_p(n)) - v(S_p(n) \cup \{n\}).
%\end{align}
%This quantity measures how much $n$ decreases the cost of the coalition consisting of her predecessors in $p$ when she joins them. Let $y_n$ be the marginal contribution of $n$ with respect to $p$, i.e., set $y_n = \Delta_p(n)$, then $\bd{y}$ is in the core of $G$. By making $p$ describes the set of permutation $\mc{P}_{\mc{N}^G}$, we can expect to simulate a part of the core of the cost game $G$.

%\begin{remark}
%This idea might be interesting to apply only when a large number of DSOs is considered.
%\end{remark}

%\subsubsection{Linear programming approach}

The constraints imposed on $\mc{C}(G)$ ensure that 
%no coalition would have an incentive to split from the grand coalition $\mc{N}^G$ and do better on its own. In other words, 
no TSO or DSO has an incentive to leave the grand coalition (of all SOs in a common market) and form any subcoalition (a sub-common market as defined in Remark~\ref{Re:Sub-common}, encompassing any subset of SOs including the singleton coalitions/disjoint markets).

We next prove and analyze the stability of the grand coalition in our game and, hence, the naturally arising common benefit of the SOs to cooperate in a common flexibility market.\vspace{-0.2cm}  
%
%\begin{definition}[Absolute Excess]
%The absolute excess of $G$ is $d(G,v) \triangleq v(\mc{N}^G) - \bd{y}_{LP}$.
%\end{definition}
%
%From the literature on production games, the following result states that the core of the cost game $G$ is not empty for a large enough number of DSOs.
%
%\begin{proposition}[\cite{flam}, Th. 4]
%The relative excess $d(G,v)$ tends to zero as $N^G \rightarrow +\infty$.
%\end{proposition}

\section{Stability of Cost Game $G$}\label{sec:Stability}
For the common market to naturally arise and be sustained, collaboration among SOs should be naturally beneficial to all of them. This is achieved if their collaboration is stable, i.e., when the core of our formulated TSO-DSO game $G$, defined in Definition~\ref{eq:TheCore}, is non-empty. 

We next prove that the core of $G$ is non-empty and hence the TSO and DSOs have a natural incentive to collaborate and form a common market.  
To this end, we first begin by proving that it is beneficial for any DSO to collaborate with the TSO than forming a disjoint market.

%\subsection{Cost Game $G$ Stability}

\begin{proposition} \label{prop:DSO_TSO_coop}
Any DSO $n \in \mc{N}$ prefers cooperating with the TSO than remaining alone.
\end{proposition}

\begin{proof}

It is profitable for any DSO $n \in \mc{N}$ to cooperate with the TSO if and only if:\vspace{-0.4cm}

\begin{small}\begin{align} \label{eq:v_relation}
v(\{\textrm{TSO} \cup n\}) \leq v(\{n\}) + v(\{\textrm{TSO}\}).
\end{align}\end{small}\vspace{-0.4cm}

Considering the relaxed version of \eqref{eq:compact_model} which we proved to be equivalent to \eqref{eq:compact_model} under weak Slater's condition, the inequality \eqref{eq:v_relation} is equivalent to the following one:\vspace{-0.4cm}

\begin{small}\begin{align} \label{eq:v_relation2}
    &\max_{\bd{\lambda}} \min_{\bd{x}_{0}, \bd{z}_{0}, \bd{x}_n, \bd{z}_n} \Big[ \bd{c}_{0}^t \bd{x}_{0} + \bd{c}_n^t \bd{x}_n + \bd{\lambda}^t (\bd{A}\bd{x}+\bd{B}\bd{z} -\bd{d})\Big] \nonumber \\
    \leq& \max_{\bd{\lambda}} \min_{\bd{x}_{0}, \bd{z}_{0}} \Big[\bd{c}_{0}^t \bd{x}_{0} + \bd{\lambda}^t (\bd{A}\bd{x}+\bd{B}\bd{z}-\bd{d})\Big] \nonumber \\
    & \hspace{1.cm} + \max_{\bd{\lambda}} \min_{\bd{x}_n, \bd{z}_n} \Big[ \bd{c}_n^t \bd{x}_n + \bd{\lambda}^t (\bd{A}\bd{x}+\bd{B}\bd{z}-\bd{d})\Big],
\end{align}\end{small}\vspace{-0.4cm}

\noindent where $\bd{x}_{0} \in \mc{X}_{0}, \bd{x}_n \in \mc{X}_n, \bd{z}_{0} \in \mc{Z}_{0}, \bd{z}_n \in \mc{Z}_n$. Since \eqref{eq:v_relation2} holds by definition, we can conclude that any DSO $n \in \mc{N}$ has an incentive to cooperate with the TSO.
\end{proof}

We next prove that not only would any DSO benefit from collaborating with the TSO, but also any DSO has an incentive to be part of the grand coalition made of all of the $\mc{N}$ DSOs and the TSO. We start by proving an intermediate result that characterizes the cost game $G$.

\begin{proposition}  \label{prop:concavity}
The characteristic function game $G=(\mc{N}^G,v)$ is concave.
\end{proposition}

\begin{proof}
Let $C \subseteq \mc{N}^G$. We compute $v(C \cup \{n\}) - v(C) = \bd{c}_n^t \bd{x}_n^*$ where $\bd{x}_n^*$ solves \eqref{eq:coal_pb}, and observe that $v(C' \cup \{n\}) - v(C')=v(C \cup \{n\}) - v(C)$ for any coalition $C' \subset C$. The concavity property introduced in Definition \ref{def:concavity} holds. 
\end{proof}

\begin{corollary} \label{corollary:submodularity}
The characteristic function cost game $G$ is submodular.
\end{corollary}

\begin{proof}
This results follows directly from Proposition \ref{prop:concavity}. Indeed, every concave game is necessarily submodular \cite{chalkiadakis}.
\end{proof}

This, hence, allows us to prove the non-emptiness of the core, as shown next. 

\begin{theorem}  \label{theorem:core}
The core of the characteristic function cost game $G$ is non-empty and the grand coalition of $G$ is stable.
\end{theorem}

\begin{proof}
Consider the following allocation linear program:\vspace{-0.4cm}

\begin{small}\begin{subequations} \label{eq:AP}
\begin{align}
    \max_{\bd{y}} & \hspace{1.cm} \sum_{n\in\mc{N}^G} y_n, \label{eq:AP1} \\
    s.t. & \hspace{1.cm} \sum_{n\in C} y_n \leq v(C), \forall C \subseteq \mc{N}^G. \label{eq:AP2}
\end{align}
\end{subequations}\end{small}\vspace{-0.2cm}

%The outcome of \eqref{eq:AP} is denoted $\bd{y}_{LP}$. 
It is quite obvious that $\mc{C}(G) \neq \emptyset $ if and only if the optimum value of the linear program \eqref{eq:AP} is equal to $v(\mc{N}^G)$, in which case any optimal solution to \eqref{eq:AP} lies in $\mc{C}(G)$. Taking the linear program dual to \eqref{eq:AP}, an equivalent
condition for $\mc{C}(G) \neq \emptyset $ can be obtained based on the concept of balanced sets. A collection $\mc{B}$ of nonempty subsets of $\mc{N}^G$ is balanced if  $\sum_{C \in \mc{B}} \gamma_C.v(C) \leq v(\mc{N})$ holds  for every balanced collection $\mc{B}$ with weights $(\gamma_C)_{C \in \mc{B}}$. A game has a non-empty core if and only if it is balanced \cite{bondareva}.  

The cost game $G$ being submodular from Corollary \ref{corollary:submodularity}, \eqref{eq:AP} is equivalent to:\vspace{-0.4cm}

\begin{small}\begin{subequations} \label{eq:AP_eq}
\begin{align}
    \max_{\bd{y}} & \hspace{1.cm} \sum_{n\in\mc{N}^G} y_n, \\
    s.t. & \hspace{1.cm} \sum_{n\in\mc{N}^G} y_n = \sum_{n\in\mc{N}^G} \Phi_n(\bd{x}_n^*), \\
         & \hspace{1.cm} \bd{x}^* = \arg\min_{\bd{x},\bd{z}} \sum_{n\in\mc{N}^G} \Phi_n(\bd{x}_n), \\
         & \hspace{2.5cm} s.t. \hspace{0.5cm} \bd{A} \bd{x} + \bd{B} \bd{z} \leq \bd{d}, \\
         &  \hspace{3.5cm} \bd{x} \in \mc{X}, \bd{z} \in \mc{Z}.
 \end{align}
\end{subequations}\end{small}\vspace{-0.4cm}

By construction, the Shapley\footnote{The Shapley value is introduced in detail and analyzed in Section~\ref{subsec:Shapley}.} value of the cost game $G$ is solution to \eqref{eq:AP_eq} and as such, belongs to the core of the game \cite{shapley, chalkiadakis}. This implies that the core of $G$ is never empty. 
\end{proof}

Theorem \ref{theorem:core}, hence, proves that it is beneficial for all DSOs and TSO to cooperate in a common market. In the derivation, $c_n$ was used in the cost function of the players, which would reflect a cost-based market clearing or a pay-as-bid mechanism. However, Theorem \ref{theorem:core} can be readily extended beyond pay-as-bid for other pricing schemes, such as nodal pricing (i.e., locational marginal pricing) as shown next.

%We next prove that stable cooperation also holds when considering nodal pricing.

%\begin{remark} 
%In addition, the Shapley value is the center of mass of the core's vertices \cite{shapley}, the kernel is a single point corresponding to the nucleolus \cite{maschler}, and the stable set and bargaining set for the grand coalition coincide with the core \cite{maschlerV2}.
%\end{remark}

%\subsection{Impact of Pricing Schemes on Cost Game $G$ Stability}

%Previously, we have proved the non-emptiness of the core of cost game $G$ and the stability of its grand coalition, assuming pay-as-bid is implemented as pricing scheme. Indeed, $c_n$ is used in the cost function of the players. However, Theorem \ref{theorem:core} can be extended beyond pay-as-bid for other pricing schemes, such as nodal pricing.

\begin{proposition} \label{prop:pricing}
The core of the cost game $G$ is non-empty under nodal pricing.
\end{proposition}

\begin{proof}
Under nodal pricing, the cost game $G$ coalitional value becomes $v(C) = \sum_{n\in C} \bd{\lambda}_n^t \bd{x}_n^* \; \textrm{where $C \subseteq \mc{N}^G$},$ and $(\bd{x}_n^*)_{n\in C}$ is the optimum of \eqref{eq:coal_pb} where $\Phi_n(\bd{x}_n)$ is replaced with $\bd{\lambda}_n^t \bd{x}_{n}$. 
%Under averaged nodal pricing, $G$ coalitional value becomes $v(C) = \frac{1}{N}\sum_{n\in \mc{N}}\lambda_n^t.\sum_{n\in C} \bd{x}_n^* \; \textrm{where $C \subseteq \mc{N}^G$},$ and $(\bd{x}_n^*)_{n\in C}$ is the optimum of \eqref{eq:coal_pb} where $\Phi_n(\bd{x}_n)$ is replaced with $\frac{1}{N}\sum_{n\in \mc{N}}\lambda_n^t. \sum_{n \in C} \bd{x}_{n}$.

Using the same reasoning as in the proof of Proposition~\ref{prop:concavity}, it is straightforward to check that cost game $G$ is concave under both pricing schemes, rendering the core of the cost game $G$ non-empty under nodal pricing.
\end{proof}

After proving the stability of the cooperation of the SOs in the common market, we next present several cost allocation methods --  based on which the total cost of flexibility procurement in the common market can be split -- and analytically characterize their properties.

\section{Allocation Mechanisms} \label{sec:allocation_mechanisms}
We introduce several cost allocation mechanisms and study their properties for our TSO-DSOs cooperative game, %based on the following dimensions: 
%We evaluate the cost allocation possibilities and their properties for our TSO-DSOs game, 
based on efficiency, stability, and fairness criteria, which measure how well each SO allocated cost reflects its contribution to the total cost. 
%
%In this regard, we introduce several cost allocation mechanisms and study their properties for the TSO-DSOs game. 
%We introduce below two main properties, namely, dummy player and additivity of coalitional games, which will be included in the studied set of cost allocation mechanisms' properties.
%
%\begin{definition}[Dummy player]
%A player $n \in \mc{N}^G$ is said to be dummy if $v(C)=v(C\cup \{n\}), \forall C \subseteq \mc{N}^G$.
%\end{definition}
%
%\begin{definition}[Sum of coalitional games]
%Consider that the $N^G$ players are involved in two coalitional games $G$ and $\widetilde{G}$. The sum of $G$ and $\widetilde{G}$ is a coalitional game $G+\widetilde{G}$, where for any coalition $C \subseteq \mc{N}^G$ the characteristic function takes the form $v(G)+v(\widetilde{G})$.
%\end{definition}
%
The properties for evaluating a cost allocation mechanism, $\bd{\Phi}(G)$, are defined as follows:

(i) \textbf{efficiency}: $\sum_{n} \Phi_n(G) = v(\mc{N}^G)$.

(ii) \textbf{dummy player}: if a player $n$ is a dummy in $G$, i.e., $v(C \cup \{n\})=v(C), \forall C \subseteq \mc{N}^G \setminus \{n\}$, then $\Phi_n(G)=0$.

(iii) \textbf{symmetry} (equal treatment of equals): if $n$ and $n'$ are equivalent in $G$, in the sense that $v(C \cup \{n\})=v(C \cup \{n'\}), \forall C \subseteq \mc{N}^G$, then $\Phi_n(G)=\Phi_{n'}(G)$.

(iv) \textbf{additivity}: $\Phi_n(G+\widetilde{G})=\Phi_n(G)+\Phi_n(\widetilde{G}), \forall n \in \mc{N}^G$.

(v) \textbf{stability}: the cost allocation $\bd{\Phi}(G)$ belongs to the core. 

(vi) \textbf{anonymity}: players' relabeling does not affect their cost allocation. 
If $n$ and $n'$ are two players, and game $\widetilde{G}$ is identical to $G$ except for exchanging the roles of $n$ and $n'$, then $\Phi_n(G)\textrm{$=$}\Phi_{n'}(\widetilde{G})$.
Note that (vi) implies (iii).

The studied cost allocation mechanisms are defined next.

\subsection{Shapley Value (SV)}\label{subsec:Shapley}

The Shapley value is a solution concept typically formulated with respect to the grand coalition: it defines a way of distributing the value $v(\mc{N}^G)$ that could be obtained by the grand coalition. The SV is based on the intuition that the cost allocated to each agent (in our case to each SO) should be \emph{proportional to her contribution} to the grand coalition. Define $\mc{P}_{\mc{N}^G}$ as the set of permutations, e.g., one-to-one mappings from $\mc{N}^G$ to itself. We introduce $S_p(n)$ as the set of all the predecessors of $n$ in $p \in \mc{P}_{\mc{N}^G}$, i.e., $S_p(n)\triangleq \{n'\in\mc{N}^G|p(n')<p(n)\}$ where $<$ denotes the predecessor relationship. The SV of SO $n$ is denoted $\textrm{SV}_n(G)$ and is given by $\textrm{SV}_n(G) \triangleq \frac{1}{N^G}\sum_{p \in \mc{P}_{\mc{N}^G}} \Delta_p(n)$ where $\Delta_p(n)\triangleq v\big(S_p(n)\cup \{n\}\big)-v\big(S_p(n)\big)$ measures the marginal contribution of $n$ with respect to a permutation $p$. It can equivalently be written under the extended form: $\textrm{SV}_n(G) = \sum_{C \subseteq \mc{N}^G \setminus \{n\}} \frac{|C|!(N^G-|C|-1)!}{N^G!} \Big[ v(C\cup \{n\})-v(C)\Big]$.

For each permutation (ordering) of the SOs, each SO is imputed a cost based on how much the SO contributes to the coalition formed by its predecessors in this permutation. The allocated cost is averaged over all possible permutations to guarantee the symmetry of the allocation.

By construction, $\bd{\textrm{SV}}(G)$ meets properties (i)-(iv)~\cite{shapley}. In fact, the SV is the \emph{only} cost allocation method that has the four properties (i)-(iv) simultaneously. The anonymity property (vi) is also met by the SV, meaning that the SV does not discriminate between the SOs solely on the basis of their indices~\cite{hart}. Finally, following the proof of Theorem~\ref{theorem:core}, property (v) holds true for our cost game $G$. The main challenge of the SV is in its computational complexity (NP - complete).  
%e.g., every exact algorithm for SV computation requires an exponential number of operations, 
%but it can be approximated efficiently relying on sampling \cite{castro}. 
The complexity of the SV also hinders the interpretation of its fairness.

%\kern-1.35em % blank space reduction

\subsection{Normalized Banzhaf Index (B\textsuperscript{\#})}

Like the SV, the Banzhaf index $\bd{\textrm{B}}(G)$ measures the agents' expected marginal contributions; however, instead of averaging over all permutations of players, it averages over all coalitions in the game. The Banzhaf index of an SO $n$ is denoted $B_n(G)$ and is given by $B_n(G) \triangleq \frac{1}{2^{N^G-1}}\sum_{C \subseteq \mc{N}^G\setminus \{n\}} \Big[ v(C\cup \{n\})-v(C)\Big]$.
$\bd{\textrm{B}}(G)$ meets properties (ii)-(iv) \cite{chalkiadakis}. Similarly to the SV, it meets also the anonymity property (vi). Because it lacks the efficiency property (i), $\bd{\textrm{B}}(G)$ is not in the core.
To meet the efficiency property (i), a rescaled version of the Banzhaf index has been proposed, called the normalized Banzhaf index: $B_n^{\sharp}(G) \triangleq \frac{B_n(G)}{\sum_{n'\in\mc{N}^G}B_{n'}(G)}$.
The normalized Banzhaf index meets properties (i)-(iii); however it loses (iv) \cite{chalkiadakis}. Similarly to Banzhaf index, it meets property (vi). We next prove that the normalized Banzhaf index leads to a stable cost allocation. 

\begin{proposition} \label{prop:Banzhaf_convex_game}
The normalized Banzhaf index $\bd{B}^{\sharp}(G)$ meets the stability property (v) for the cost game $G$.
\end{proposition}

\begin{proof}
By construction, the normalized Banzhaf index is efficient. We need to check the non-deviation property, i.e., that $\sum_{n\in C} B_n^{\sharp}(G) \leq v(C), \forall C \subseteq \mc{N}^G$. We notice that the normalized Banzhaf index is a convex combination of elements that are constructed in the proof of Th.2.27 in~\cite{chalkiadakis} to exhibit elements from the core of $G$. Since these elements belong to the core by construction, the normalized Banzhaf index is itself a convex combination of these elements, and the core of the cost game $G$ can be shown to be a convex set. Hence, the normalized Banzhaf index meets the non-deviation property, which implies that it belongs to the core of $G$.
\end{proof}

Similarly to the SV, both versions of Banzhaf index share exponential computational complexity rates, implying challenges for implementations on a large scale and no simple interpretation of fairness.

\subsection{Cost Gap Allocation Method (CGA)}

This method coincides with the $\tau$-value, introduced in
%by Tijs and Driessen 
~\cite{tijs}. Similarly to the SV, we define $\bd{\Delta} = \big(\Delta(n)\big)_n$ as the marginal cost vector. Its $n$-th coordinate is the separable cost of SO $n$, i.e., $\Delta(n) \triangleq v(\mc{N}^G)-v(\mc{N}^G\setminus \{n\}), \forall n \in \mc{N}^G$. 
For each coalition $C \subseteq \mc{N}^G$, we define the cost gap of $C$ by: \vspace{-0.4cm}

\begin{small}\begin{align*}
\left\{
    \begin{array}{ll}
        g(C)  \triangleq  \hspace{0.2cm} v(C) - \sum_{n\in C} \Delta(n) & \textrm{if $C\neq \emptyset$}, \\ 
        g(\emptyset)  \triangleq  0.
        \end{array}
\right.
\end{align*}\end{small}\vspace{-0.4cm}

%under the standard assumption that the value of the empty coalition $\emptyset$ is $0$. 
The map $g:2^{\mc{N}^G} \rightarrow \mathbb{R}$ is the cost gap function of game $G$. Note that $g(\mc{N}^G)$ is equal to the non separable cost in $G$. In general, we consider $g \geq 0$. 
We define the weight vector $\bd{w}$ such that: $w_n \triangleq \min_{\{C | n \in C\}} g(C), \forall n \in \mc{N}^G$.
For any characteristic function $v$ such that $g(C) \geq 0, \forall C \subseteq \mc{N}^G$ and $\sum_{n\in\mc{N}^G} w_n \geq g(\mc{N}^G)$, the cost gap allocation method assigns the cost allocation:\vspace{-0.4cm}

\begin{small}\begin{align*}
\bd{y} \triangleq & 
\left\{
    \begin{array}{ll}
        \bd{\Delta} & \textrm{if $g(\mc{N}^G)=0$}, \\
        \bd{\Delta} + g(\mc{N}^G)(\sum_{n\in\mc{N}^G}w_n)^{-1}\bd{w} & \textrm{if $g(\mc{N}^G)>0$.}
    \end{array}
\right.
\end{align*}\end{small}\vspace{-0.4cm}

The cost gap allocation $\bd{\textrm{CGA}}(G)$ meets the efficiency (i), dummy player (ii), anonymity (vi) and, therefore, symmetry (iii) properties \cite{tijs}. $\bd{\textrm{CGA}}(G)$ is stable for $N^G<4$, but can lead to unstable outcomes for $4 \geq N^G$ \cite{tijs}. 
On the other hand, $\bd{\textrm{CGA}}(G)$ gives rise to an exact analytical expression, and therefore to a simple interpretation of fairness.

\subsection{Lagrangian Based Allocation (L)}
%\tcb{This is actually based on the dual function. Since this is a linear problem, the object value is equal to the sum of the dual of a constraint multiplied by RHS. The cost allocation is calculated according to which agent the RHS belongs to.}

This method is an extension of the classical shadow price based cost split \cite{frisk}, with which it coincides when weak duality holds for the cost game $G$. In \eqref{eq:compact_model}, we get the dual $\bd{\lambda}$ for the coupling constraint \eqref{eq:compact2}, and duals $\bd{\mu}_0$, $(\bd{\mu}_n)_n$ for the individual constraints \eqref{eq:compact3}, \eqref{eq:compact4}. When solving \eqref{eq:compact_model} for the grand coalition, we get $v(\mc{N}^G)$. The optimal dual solution has the property that:\vspace{-0.4cm}

\begin{small}\begin{equation} \label{eq:shadow_price}
v(\mc{N}^G)\textrm{=}\,\bd{\lambda}^t (\bd{A}\bd{x}\textrm{$+$}\bd{B}\bd{z}\textrm{$-$}\bd{d})\textrm{$+$}\bd{\mu}_0^t\Psi_0(\bd{x}_0)\!\textrm{$+$}\!\!\sum_n \!\bd{\mu}_n^t\Psi_n(\bd{x}_n).
\end{equation}\end{small}\vspace{-0.4cm}

Since our problem is linear, each SO's contribution can be found by computing its contribution to the dual objective function value \cite{frisk}. This cost allocation is efficient (i) under weak duality \cite{owen}. Stability (v) is also achieved since efficiency (i) holds and the individual rationality condition is met by definition of the cost allocation and weak duality. From \eqref{eq:shadow_price}, dummy player (ii) and additivity (iv) hold. However, symmetry (iii) and anonymity (vi) do not hold in general. 

Lagrangian based cost allocation $\bd{\textrm{L}}(G)$ requires the computation of the grand coalition optimal value (e.g., solving a linear optimization program) and all the dual variables associated with the constraints in \eqref{eq:compact_model}. Making use of a solver, $\bd{\textrm{L}}(G)$ implementation is simple and leads to a rather intuitive interpretation of fairness.

\subsection{Equal Profit Method (EPM)}

The motivation behind this cost allocation is to propose a method which is aimed at finding a stable cost allocation, such that the maximum difference in pairwise relative savings is minimized. We call this method, the Equal Profit Method (EPM)~\cite{frisk}. The relative savings of SO $n$ is, then, computed as $\frac{v(\{n\})-y_n}{v(\{n\})} = 1- \frac{y_n}{v(\{n\})}$. When a cost allocation is stable, $v(\{n\}) \geq y_n$. This allocation is obtained using the solution of a linear optimization problem:\vspace{-0.4cm}

\begin{small}\begin{subequations}
\begin{align}
\min & \hspace{1.cm} f, \\
s.t. & \hspace{1.cm} f \geq \frac{y_n}{v(\{n\})}-\frac{y_{n'}}{v(\{n'\})}, \forall n,n', \label{eq:EP0} \\
& \hspace{1.cm} \sum_{n' \in C} y_{n'} \leq v(C), \forall C \subseteq \mc{N}^G, \label{eq:EP1} \\
& \hspace{1.cm} \sum_{n'\in\mc{N}^G} y_{n'} = v(\mc{N}^G), y_n \geq 0, \forall n. \label{eq:EP2}
\end{align}
\end{subequations}\end{small}\vspace{-0.2cm}

EPM belongs to the core of the cost game $G$ by construction, as \eqref{eq:EP1}, \eqref{eq:EP2} define efficiency and individual rationality\footnote{A cost allocation vector $y$ for a coalition structure $CS$ satisfies the individual rationality property if $y_n \leq v(\{n\}), \forall n \in \mc{N}^G$, i.e., each SO weakly pefers being in the coalition structure to being on his own.} respectively. This means that it meets the efficiency (i) and stability (v) properties. The symmetry property (iii) is naturally met from \eqref{eq:EP0} at the optimum. However, additivity (iv), anonymity (vi) and dummy player (ii) are not met in general.  
EPM requires the optimal solution of a linear constrained optimization problem. Making use of a solver, its implementation is simple and leads to an intuitive interpretation of fairness.

\subsection{Proportional Cost Allocation (PCA)}

A straightforward allocation is to distribute the total cost of the grand coalition (i.e., the total cost of the common market), $v(\mc{N}^G)$, among the SOs according to how much flexibility is used by each of them. This is expressed by $y_n = w_n.v(\mc{N}^G)$, where $w_n$ is equal to SO $n$'s share of the total activated flexibility, and $w_n\geq 0, \forall n \in \mc{N}^G, \sum_{n'\in\mc{N}^G}w_{n'}=1$, or, alternatively in case stand-alone costs are used, it is equal to $\frac{v(\{n\})}{\sum_{n'\in\mc{N}^G}v(\{n'\})}$ \cite{frisk}. 

The proportional cost allocation method is easy to understand, implement, and compute. Besides, it is stable if and only if $\sum_{n\in C} w_n \leq \frac{v(C)}{v(\mc{N}^G)}, \forall C \subseteq \mc{N}^G$. In case stand-alone costs are used (we refer to this specifically as PCA), the proportional allocation is stable, meeting (v). The efficiency property (i) is always met due to the normalization of the weights. The symmetry property (iii) also holds, but neither anonymity (vi) nor dummy player (ii) hold. 

As PCA gives rise to an exact analytical expression, it yields a simple, intuitive interpretation of fairness.

A comparison of the properties of the different proposed cost allocation methods is provided in Table \ref{tab:comparison_allocations}.

%%%
\begin{table*}%\vspace{-0.2cm}
 \caption{{\small Properties of the Cost Allocation Methods for cost game $G$.}}
\label{tab:comparison_allocations}
\begin{tabularx}{\textwidth}{@{}l*{10}{C}c@{}}
\toprule
Properties & \textbf{SV} & \textbf{$\bd{B}^{\sharp}$} & \textbf{CGA} & \textbf{L} & \textbf{EPM} & \textbf{PCA} \\\hline
Efficiency & $\checkmark$ & & $\checkmark$ & $\checkmark$ & $\checkmark$ & $\checkmark$ \\
Dummy player & $\checkmark$ & $\checkmark$ & $\checkmark$ & $\checkmark$ & &  \\
Symmetry & $\checkmark$ & $\checkmark$ & $\checkmark$ &  & $\checkmark$ & $\checkmark$ \\
Additivity & $\checkmark$ & $\checkmark$ & & $\checkmark$ & & $\checkmark$ \\
Stability & $\checkmark$ & & $\checkmark$ if $N^G<4$ & $\checkmark$ & $\checkmark$ & $\checkmark$ \\
Anonymity & $\checkmark$ & $\checkmark$ & $\checkmark$ & & & \\\hline
Complexity & NP-complete & NP-complete & exact & $O(N^G)$ & $O(N^G)$ & exact \\\hline
Simplicity & & & $\checkmark$ & $\checkmark$ & $\checkmark$ & $\checkmark$ \\
\bottomrule
\end{tabularx}
\end{table*}
%%%

\section{Numerical Results}\label{sec:NumericalResults}

%Setting: one TSO, three DSOs 

%To show: 
%\begin{itemize}
%\item First to simulate the incremental benefit of cooperation up to the grand coalition. We start by treating each system operator alone, then incrementally add one DSO to the coalition and compute the total costs.
%\item costs of the TSO, DSOs in each cost allocation method \\
%\item sensitivity analysis with respect to the limits on $T_p$. What is the benefit of allowing more exchange between the different systems.
% \item wording for three DN18
% \item compare with Shapley
%\item Comparing cost allocation using Shapley value under different pricing schemes
%\item variances of the DSOs cost allocations \\
%\item SC to compare Banzhaf to the other cost allocation methods \\
%\item loss induced when implementing one Shapley value based on sampling \cite{castro}
%\item congestion costs for TSO/DSO
%\end{itemize}

For the numerical results, we consider an interconnected system composed of the IEEE 14-bus transmission network connected to three distribution networks -- namely, the Matpower 18-bus, 69-bus, and 141-bus systems. 
All network parameters are based on the corresponding cases in Matpower~\cite{matpower}. 
We add base demand to the buses and adapt the capacity limit of lines, so that the initial system state (without any flexibility activation) show anticipated congestions and system imbalance. 
%the transmission network is both congested and unbalanced, while the three distribution systems show congestions. 
Flexibility bids for both upward and downward flexibility are created over the different nodes. The submitted bids are drawn from a uniform distribution in the range of $[10, 15]$\euro/MWh for downward bids and $[50, 55]$ \euro/MWh for upward offers, not to induce biases stemming from the submitted bids/offers.
%following a similar approach as in \cite{CIRED21}. 
%The numerical simulation results are presented as follows.

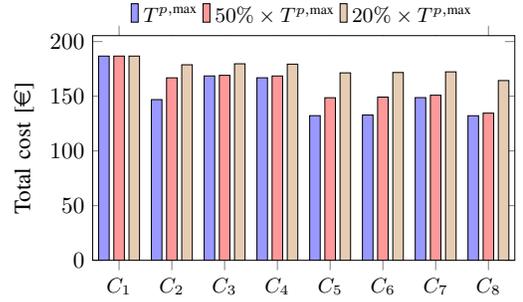
\begin{figure}
	\centering
	\resizebox{0.8\columnwidth}{!}{
		\begin{tikzpicture}

\begin{axis}[
    ylabel={Total cost [\euro]},
    xtick = {0,1,2,3,4,5,6,7},
    xticklabels={$C_1$, $C_2$, $C_3$, $C_4$, $C_5$, $C_6$, $C_7$, $C_8$},
    ymin=0,
    xmin=-0.5,
    xmax=7.5,
    ybar,
    width=8cm, 
    height=5cm,
    bar width=0.2,
    legend style={
    	font=\small, 
    	draw=none,
    	legend columns=3,
    	at={(0.5,1)},
    	anchor=south
    },
	legend cell align={left},
    every x tick label/.style={font=\small},
%	every node near coord/.style={font=\tiny},
]

    \addplot[fill=blue!40] table[x=coalition, y=ResultFirst, col sep=semicolon]{figures/incremental_benefit_data.csv};
    \addlegendentry{$T^{p,\textrm{max}}$}
    \addplot[fill=red!40] table[x=coalition, y=ResultSecond, col sep=semicolon]{figures/incremental_benefit_data.csv};
    \addlegendentry{$50\%\times T^{p,\textrm{max}}$}
    \addplot[fill=brown!40] table[x=coalition, y=ResultThird, col sep=semicolon]{figures/incremental_benefit_data.csv};
    \addlegendentry{$20\%\times T^{p,\textrm{max}}$}
\end{axis}

\end{tikzpicture}
	}
	\caption{{\small Total cost reduction when different coalitions are formed in three cases of interface flow limit, where:} {\small $C_1$=\big\{\{TN\}, \{DN$_{18}$\}, \{DN$_{69}$\}, \{DN$_{141}$\}\big\}, $C_2$=\big\{\{TN, DN$_{18}$\}, \{DN$_{69}$\}, \{DN$_{141}$\}\big\}, $C_3$=\big\{\{TN, DN$_{141}$\}, \{DN$_{18}$\}, \{DN$_{69}$\}\big\}, $C_4$=\big\{\{TN, DN$_{69}$\}, \{DN$_{18}$\}, \{DN$_{141}$\}\big\}, $C_5$=\big\{\{TN, DN$_{18}$, DN$_{69}$\}, \{DN$_{141}$\}\big\}, $C_6$=\big\{\{TN, DN$_{18}$, DN$_{141}$\}, \{DN$_{69}$\}\big\}, $C_7$=\big\{\{TN, DN$_{69}$, DN$_{141}$\}, \{DN$_{18}$\}\big\}, $C_8$=\big\{\{TN, DN$_{18}$, DN$_{69}$, DN$_{141}$\}\big\}}.\vspace{-0.4cm}  		 %The interface flow limit in the case of $TP_2$ and $TP_3$ is 50\% and 20\% of that in the case of $TP_1$, respectively.
	}
	\label{fig:incremental_benefit}
\end{figure}%\vspace{-0.6cm}

We first showcase the benefit of cooperation in reducing the total system costs. Fig.~\ref{fig:incremental_benefit} shows the incremental benefit of cooperation when adding additional DSOs up to the grand coalition as well as the impact that the interface flow limits have on the cooperation benefit. We consider first each SO alone (i.e. disjoint markets), then incrementally add one DSO to the coalition and compute the total costs. The coalition formed are denoted $C_1$ to $C_8$ (defined in Fig.~\ref{fig:incremental_benefit}), where $C_1$ represents the case of disjoint markets (i.e., singleton SOs) and $C_8$ the case of the grand coalition forming a common market. The process is repeated for three different levels of maximum interface flow limits, where the $50\%\times T^{p,\textrm{max}}$ and $20\%\times T^{p,\textrm{max}}$ cases reduce the interface flow maximum limit for each DSO to, respectively, $20\%$ and $50\%$ of its original value, $T^{p,\textrm{max}}$.  
%denoted as follows: $C_1$=\big\{\{TN\}, \{DN$_{18}$\}, \{DN$_{69}$\}, \{DN$_{141}$\}\big\}, $C_2$=\big\{\{TN, DN$_{18}$\}, \{DN$_{69}$\}, \{DN$_{141}$\}\big\}, $C_3$=\big\{\{TN, DN$_{141}$\}, \{DN$_{18}$\}, \{DN$_{69}$\}\big\}, $C_4$=\big\{\{TN, DN$_{69}$\}, \{DN$_{18}$\}, \{DN$_{141}$\}\big\}, $C_5$=\big\{\{TN, DN$_{18}$, DN$_{69}$\}, \{DN$_{141}$\}\big\}, $C_6$=\big\{\{TN, DN$_{18}$, DN$_{141}$\}, \{DN$_{69}$\}\big\}, $C_7$=\big\{\{TN, DN$_{69}$, DN$_{141}$\}, \{DN$_{18}$\}\big\}, $C_8$=\big\{\{TN, DN$_{18}$, DN$_{69}$, DN$_{141}$\}\big\}. 
As can be seen from Fig.~\ref{fig:incremental_benefit}, adding more DSOs in the cooperation further reduces the total costs. For example, by comparing $C_1$, $C_2$, $C_5$ and $C_8$, we can observe the way the sequential addition of DN$_{18}$, DN$_{69}$, and DN$_{141}$ to the TSO (TN) led to a significant decrease in total system costs for all values of interface flow limits ($C_8$ achieves 29\%, 28\%, and 12\% reduction with respect to $C_1$ for, respectively, $T^{p,\textrm{max}}$, $50\%\times T^{p,\textrm{max}}$, and $20\%\times T^{p,\textrm{max}}$). This also highlights the effect of allowing a higher level of interface power exchange on achieving a more efficient procurement of flexibility, as out of the three interface flow limits, the $20\%\times T^{p,\textrm{max}}$ case achieves the least amount of savings. 

Indeed, Fig.~\ref{fig:incremental_benefit} further showcases the way in which a tighter interface flow limit reduces the benefits introduced by cooperation, by showcasing the increased total cost under each coalition for the three limits $T^{p,\textrm{max}}$, $50\%\times T^{p,\textrm{max}}$, and $20\%\times T^{p,\textrm{max}}$. However, the impacts of the interface flow limits on different DSOs vary significantly. By comparing bars of the same color in $C_2$, $C_3$ and $C_4$, we can see that $C_2$ shows the most severe variation in costs for a tighter interface flow limit. This reflects the fact that DN$_{18}$ is affected more significantly, as in this case study, DN$_{18}$ has the largest contribution to the total cost reduction in the common market. Hence, when its flow with the TSO is more limited, the benefit from cooperation significantly decreases.
%, highlighting the need for allowing a higher capacity for power exchange with DN$_18$. 
%This statement can be confirmed by comparing bars of the same color in $C_1$ to $C_2$, $C_3$ and $C_4$. 

In addition to the flexibility bids (prices and quantities) submitted from a given distribution network, the location of the distribution network itself within the system (i.e., the node in the transmission system to which the distribution network is connected) also plays a key role in its contribution to the cost reduction of the coalition. For example, we consider a case in which we place three identical 18-bus distribution networks including the same set of bids (denoted DN$_A$, DN$_B$, and DN$_C$) at different buses of the transmission network and evaluate the total costs when each DSO cooperates with the TSO. The numerical results are shown in Table~\ref{tab:cost_comparison}. Even though the distribution systems and the set of submitted bids from those systems are the same, the resulting cooperation-induced reduction in system costs are different -- as shown Table~\ref{tab:cost_comparison} -- due to the location of transmission system congestions.
%the accepted bids located at the distribution network in each case are the same, 
%the accepted bids from TSO-level flexibility are different, due to the existence of congestion in the transmission network, leading to different system costs as shown in Table~\ref{tab:cost_comparison}. 
As shown Table~\ref{tab:cost_comparison}, the total cost is reduced by 10.2\% when the TSO cooperates with DN$_B$ (third row in Table~\ref{tab:cost_comparison}), but this saving drops to only 6.4\% when the TSO cooperates with DN$_C$ (fourth row in Table~\ref{tab:cost_comparison}).

\begin{small}\begin{table}
	\caption{Total costs when identical DSOs located at different buses cooperate with the TSO.}
	\label{tab:cost_comparison}
	\centering
	\begin{tabular}{l c} 
		\toprule
		Coalition & Total Costs [\euro{}] \\ 
		\hline
		\big\{\{TN\}, \{DN$_{A}$\}, \{DN$_{B}$\}, \{DN$_{C}$\}\big\} & 43.15  \\ 
		\big\{\{TN, DN$_{A}$\}, \{DN$_{B}$\}, \{DN$_{C}$\}\big\} & 39.81  \\ 
		\big\{\{TN, DN$_{B}$\}, \{DN$_{A}$\}, \{DN$_{C}$\}\big\} & 38.74  \\ 
		\big\{\{TN, DN$_{C}$\}, \{DN$_{A}$\}, \{DN$_{B}$\}\big\} & 40.4 \\
		\hline
	\end{tabular}\vspace{-0.3cm}
\end{table}\end{small}

%The total cost is reduced from 43.15\euro{} when the TSO does not cooprate with any DSOs, to 39.81\euro{} when the TSO cooperates only with DSO A, 38.74\euro{} when the TSO cooperates only with DSO B, and 40.4\euro{} when the TSO cooperates only with DSO C,  respectively.

%\tcb{Fig. \ref{fig:incremental_benefit}}
%\begin{itemize}
%	\item Adding more DSOs in the coalition help to reduce the total cost.
%	\item The limit on interface flow also limits the benefit of coalitions.
%	\item The impact of the interface flow limit on different DSOs are not the same.
%		\item The contribution of DSOs varies. (How, why: conjecture -  location of DSO and price/quantity of bids)
%\end{itemize}

\begin{figure}
	\centering
	\begin{subfigure} {0.8\columnwidth}
		\centering
		\resizebox{\columnwidth}{!}{
	
			\begin{tikzpicture}

\begin{axis}[
    ylabel={Cost [\euro]},
    xtick = {0,1,2,3,4,5},
    xticklabels={SV, B, CGA, PCA, EMP, L},
	xmajorticks=false,
    ymin=0,
    xmin=-0.5,
    xmax=5.5,
    ybar stacked,
    width=8cm, 
    height=5cm,
    bar width=0.8,
    nodes near coords,
    legend style={
    	font=\small, 
    	draw=none,
    	legend columns=4,
    	at={(0.5,1)},
    	anchor=south
    },
	legend cell align={left},
    every x tick label/.style={font=\small, rotate=0,anchor=east},
	every node near coord/.style={font=\tiny},
]

    \addplot[fill=blue!30] table[x=method, y=TN, col sep=semicolon]{figures/cost_allocation_tp_one_data.csv};
    \addlegendentry{TN}
    \addplot[fill=green!30] table[x=method, y=DN_141, col sep=semicolon]{figures/cost_allocation_tp_one_data.csv};
    \addlegendentry{DN$_{141}$}
    \addplot[fill=red!30] table[x=method, y=DN_69, col sep=semicolon]{figures/cost_allocation_tp_one_data.csv};
    \addlegendentry{DN$_{69}$}
    \addplot[fill=yellow!50] table[x=method, y=DN_18, col sep=semicolon]{figures/cost_allocation_tp_one_data.csv};
    \addlegendentry{DN$_{18}$}
\end{axis}

\end{tikzpicture}
		}
		\vspace{-1.8\baselineskip}
		\caption{\footnotesize $T^{p,\textrm{max}}$}
		\label{fig:cost_allocation_tp_one}
	\end{subfigure}
	\vspace{0.5\baselineskip}
\begin{comment}~
	\begin{subfigure} {0.8\columnwidth}
		\centering
		\resizebox{\columnwidth}{!}{
			
			\input{figures/cost_allocation_tp_two}
		}
		\vspace{-1.8\baselineskip}
		\caption{\footnotesize 50\% $T^{p,\textrm{max}}$}
		\label{fig:cost_allocation_tp_two}
	\end{subfigure}
		\vspace{0.5\baselineskip}
~\end{comment}
	\begin{subfigure} {0.8\columnwidth}
		\centering
		\resizebox{\columnwidth}{!}{
			
			\begin{tikzpicture}

\begin{axis}[
    ylabel={Cost [\euro]},
    xtick = {0,1,2,3,4,5},
    xticklabels={\textbf{SV}, \bf{B\textsuperscript{\#}}, \textbf{CGA}, \textbf{PCA}, \textbf{EMP}, \textbf{L}},
    ymin=0,
    xmin=-0.5,
    xmax=5.5,
    ybar stacked,
    width=8cm, 
    height=5cm,
    bar width=0.8,
    nodes near coords,
    legend style={
    	font=\tiny, 
    	draw=none,
    	legend columns=4,
    	at={(0.5,1)},
    	anchor=south
    },
	legend cell align={left},
    every x tick label/.style={font=\small},
	every node near coord/.style={font=\tiny},
]

    \addplot[fill=blue!30] table[x=method, y=TN, col sep=semicolon]{figures/cost_allocation_tp_three_data.csv};
%    \addlegendentry{TN}
    \addplot[fill=green!30] table[x=method, y=DN_141, col sep=semicolon]{figures/cost_allocation_tp_three_data.csv};
%    \addlegendentry{DN$_{141}$}
    \addplot[fill=red!30] table[x=method, y=DN_69, col sep=semicolon]{figures/cost_allocation_tp_three_data.csv};
%    \addlegendentry{DN$_{69}$}
    \addplot[fill=yellow!50] table[x=method, y=DN_18, col sep=semicolon]{figures/cost_allocation_tp_three_data.csv};
%    \addlegendentry{DN$_{18}$}
\end{axis}

\end{tikzpicture}
		}
		\vspace{-1.5\baselineskip}
		\caption{\footnotesize $20\%\times T^{p,\textrm{max}}$}
		\label{fig:cost_allocation_tp_three}
	\end{subfigure}\vspace{-0.2cm}
	\caption{SOs' costs under different cost allocation methods and interface flow limit.}\vspace{-0.6cm}
	\label{fig:cost_allocation}
\end{figure}
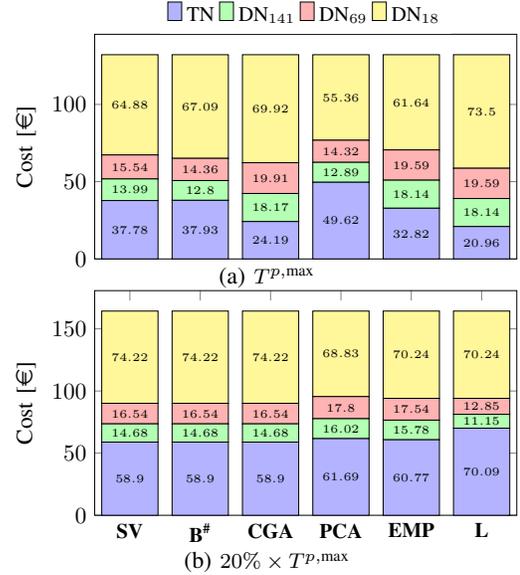

We next investigate the resulting cost for each SO under the different cost allocation methods. 
%We also consider in this analysis the effects of the maximum limit on the interface flows. 
%The results are shown in Fig.~\ref{fig:cost_allocation}. 
The stacked bars in Fig.~\ref{fig:cost_allocation} showcase the cost allocated to each system operator via the different mechanisms presented in Section~\ref{sec:allocation_mechanisms} and summarized in Table~\ref{tab:comparison_allocations}.  Fig.~\ref{fig:cost_allocation} highlights the varying proportions of the total cooperation cost to be taken up by each SO under different cost allocation methods. For example, in the results of Fig.~\ref{fig:cost_allocation}, considering the nominal $T^{p,\textrm{max}}$ case, the Shapley value leads all DSOs to bear a higher portion of the total cost as compared to, e.g., the PCA method. The results show that the SV, in this numerical setting, induces a shift in total cost from the TSO to the DSOs as compared to the PCA. 
Under this setting, the CGA is more favorable to the TSO as compared to the SV, while the PCA is less favorable, which is the opposite case to that of the DSOs. %is that each cost allocation method results in a different proportion of the total cost of the cooperation which has to be payd by each SO. 
%
%However, the SV highlights the true marginal contribution of each SO to the grand coalition and splits the total cost accordingly. 
%Hence, in this numerical setting, the disparity of cost allocation may be necessary. However, our developed method allows identifying this disparity beforehand which enables remedial actions.  
%
%can result in a disparate proportion of the total cost that will be the cost allocation method will have a direct impact on the propotion of costs that will be carried out 
%
In addition, the EMP, for example, can help reducing the cost allocated to DN$_{18}$, who, in general, bears the highest cost in all the allocation methods. However, the CGA and EPM do not cover some of the essential properties for adequate cost allocation (as shown in Table~\ref{tab:comparison_allocations}).
%Beyond the properties analyzed in Section~\ref{sec:allocation_mechanisms} and summarized in Table~\ref{tab:comparison_allocations},
%In the Lagrangian based  cost allocation method, TN, DN$_{141}$ and DN$_{69}$ have the same costs in Fig. \ref{fig:cost_allocation_tp_one} and Fig. \ref{fig:cost_allocation_tp_two}. This is due to the fact that when the interface flow limit drops by 50\%, the associated shadow prices of DN$_{141}$ and DN$_{69}$ remain zero.
%

The interface flow limits have also a direct effect on the disparity of the costs borne by each SO. For example, as shown in Fig.~\ref{fig:cost_allocation}, the 20\%$\times T^{p,\textrm{max}}$ limit increases the cost borne by the TSO for all cost allocation methods, while this limit helps reducing the cost borne by DN$_{18}$ under the Lagragian-based allocation method. In general, Fig.~\ref{fig:cost_allocation} shows that the tighter the interface flow limits the lower the difference in costs under the different cost allocation methods as cooperation becomes more limited and, hence, less consequential. 

The presented numerical results are specific to our numerical case analysis. However, they serve to highlight the fact that different cost allocations and interface flow limits can lead to disparity in the cost borne by each SO. Hence, choosing an adequate allocation scheme (meeting key properties in Table~\ref{tab:comparison_allocations}) while reducing this disparity,
%along with the characterization of the properties of each allocation method (as summarize in Table~\ref{tab:comparison_allocations}), 
are essential to achieving beneficial, stable, efficient, and fair cooperation among SOs for the joint procurement of flexibility resources.\vspace{-0.2cm}  
\section{Conclusion}\label{sec:Conclusion}
In this paper, we have introduced a common market model for the joint TSO-DSOs procurement of flexibility. We have then developed a cooperative game approach to analyze the cooperation potential of SOs in the common market. In this regard, we have proven the stability of their cooperation, implying that cooperation in this common market can naturally arise without external intervention. In addition, we have analyzed several possible cost allocation mechanisms, to split the cost of the jointly procured flexibility among the participating SOs in a stable and adequate manner, while analyzing the properties of each of those methods. Our numerical results have further highlighted the benefits of this cooperation. In addition, the results have shown the paramount effect of the interface flow limits on the benefits of cooperation, highlighting the need for further investments to improve power exchange capabilities among the different grids. The numerical results also highlighted the disparity that can be introduced by different cost allocation methods, where some methods lead to a shift in the costs borne by different SOs for the same market.%\vspace{-0.2cm}
%\vspace{-0.4cm}
\def\baselinestretch{0.92}
\bibliographystyle{IEEEtran}
\bibliography{references}

\appendix[Compact LP formulation of the common market] \label{appendix}

%We start by introducing the notation we need to formulate the problem in the compact form of Section II-D. Then, for both transmission and distribution sides, we detail how we arrive to the linear constraints formulation in $\bd{x}$ and $\bd{z}$.

%\section*{Notation}

We aim at expressing (22)-(24) in the compact linear matrix form in (25).
To that purpose, we introduce $\bd{M}_{adj}^{n^T} \in \mc{M}at(\mc{N}^{n^T},\mc{N}^{n^T})$ as the adjacency matrix of graph $\mathcal{G}^{n^T}(\mathcal{N}^{n^T},\mathcal{L}^{n^T})$, where $\bd{M}_{adj}^{n^T}(i,j)=1$ if there exists a link connecting node $i \in \mc{N}^{n^T}$ to node $j \in \mc{N}^{n^T}$; $0$ otherwise. Similarly, we define $\bd{M}_{adj}^{T} \in \mc{M}at(\mc{N}^{T},\mc{N}^{T})$ as the adjacency matrix of $\mathcal{G}^{T}(\mathcal{N}^{T},\mathcal{L}^{T})$, where $\bd{M}_{adj}^T(i,j)=1$ if there exists a link connecting node $i \in \mc{N}^T$ to node $j \in \mc{N}^T$; $0$ otherwise. In addition, we consider the generation shift factor (GSF) matrix $\bd{\Xi} \in \mc{M}at(\mc{L}^{T},\mc{N}^{T})$ such that $\bd{\Xi}\big((i,j),n\big)=X_{(i,j),n}$ is the GSF of line $(i,j) \in \mc{L}^T$ and node $n \in \mc{N}^T$. Moreover, for notation simplicity, we introduce vector $\bd{1}_{n^T}$ of size $1 \times N^T$, with $\bd{1}_{n^T}(n^T)=1$; $0$ otherwise. Similarly, for any $i \in \mc{N}^{n^T}$, we introduce vector $\bd{1}_i^{n^T}$ of size $1 \times N^{n^T}$, with $\bd{1}^{n^T}_i(i)=1; 0$ otherwise. The square matrix with vector $\bd{x}$ on its main diagonal is denoted $\textrm{diag}(\bd{x})$. $\bd{M}(n,:)$ denotes the $n$-th row of matrix $\bd{M}$, while $\bd{M}(:,n)$ denotes its $n$-th column.

%Finally, $\bd{x}^t$ denotes the transpose of vector $\bd{x}$.

%\section*{Transmission side}

We start with the \emph{transmission side}. It is clear that (2) is linear in $\bd{x}_0$ and $\bd{z}$. It can be written under the compact form: $\big( \bd{1}_{n} \; -\bd{1}_{n} \; \bd{1}_{n} \; -\bd{1}_{n}\big) \bd{x}_0 - \big( \bd{0} \; \bd{1}_{n} \; \bd{0} \big) \bd{z} = d^{T,o}_{n}-p_{n}^{T,o}, 
 \forall n \in \mc{N}^T$. Then, (3) is linear in $\bd{z}$ leading to the compact linear formulation: $- \bd{p}^{T^t} \bd{\Xi}(l,:)^t + \bd{P}^{T}(i,j) +\bd{T}^p \bd{\Xi}(l,:)^t=0, \forall l=\{i,j\} \in \mc{L}^T$. Equation (4) applies to nodes on the transmission network, excluding nodes at the interface with the distribution networks. Equation (4) is also linear in $\bd{z}$, taking the form: $p^T_n-\bd{M}_{adj}^{{T}}(n,:) (\bd{P}^T(n,:))^t = 0, \forall n \in \mc{N}^T \setminus \mc{N}^D$. Equation (5) focuses on nodes at the interface between the transmission and the distribution networks. It takes the compact form below, which is again linear in $\bd{z}$: $p^T_n-T^p_n-\bd{M}_{adj}^{{T}}(n,:) (\bd{P}^T(n,:))^t = 0, \forall n \in \mc{N}^D$. Finally, (6)--(8) provides lower and upper-bounds on the variables $\bd{P}^T(i,j)$ in $\bd{z}$ and on $\bd{x}_0$, respectively. 

%\section*{Distribution side}

At the \emph{distribution side}, we consider any DSO--$n^T \in \mc{N}^{n^T}$. 
(10) is linear in $\bd{x}_{n^T}$ and $\bd{z}$: $\Big( -\bd{1}_n^{n^T} \; \bd{1}_n^{n^T} \; -\bd{1}_n^{n^T} \; \bd{1}_n^{n^T} \; \bd{0} \Big) \bd{x}_{n^T} + \Big(\bd{0} \; \bd{1}_n^{n^T} \; \bd{0} \Big) \bd{z} 
= \bd{1}_n^{n^T} \big( \bd{p}^{n^T,o} - \bd{d}^{n^T,o}\big), \forall n \in \mc{N}^{n^T}$. (11) is linear only in $\bd{z}$, and can be written as follows: 
$\bd{p}^{n^T}_n+\bd{1}_n^{n^T}\bd{P}^{n^T}(n,:)-\bd{M}_{adj}^{n^T}(n,:)(\bd{P}^{n^T}(n,:))^t=\bd{0}$. (12) is linear in $\bd{x}_{n^T}$ and $\bd{z}$:
$\bd{q}^{n^T}_n+\bd{1}_n^{n^T}\bd{Q}^{n^T}(n,:)-\bd{M}_{adj}^{n^T}(n,:)(\bd{Q}^{n^T}(n,:))^t=\bd{0}$,
%$\bd{q}^{n^T}+(\bd{I}-\bd{M}_{adj}^{n^T})\bd{Q}^{n^T}=\bd{0}$, 
%which can be rewritten as: %$\bd{q}^{n^T}=-(\bd{I}-\bd{M}_{adj}^{n^T})\bd{Q}^{n^T}$. 
%----
(13)-(14) are linear in $\bd{z}$, giving rise to: $T_{n^T}^p-\bd{M}_{adj}^{n^T}(n_0^{n^T},:)(\bd{P}^{n^T}(n_0^{n^T},:))^t=0$ and $T_{n^T}^q-\bd{M}_{adj}^{n^T}(n_0^{n^T},:)(\bd{Q}^{n^T}(n_0^{n^T},:))^t=0$. For (15), we get: $v_n^{n^T}-v_{A(n)}^{n^T}+2. M_{adj}^{n^T}(n,:)\big(\textrm{diag}(\bd{r}^{n^T})(\bd{P}^{n^T}(n,:))^t-\textrm{diag}(\bd{x}^{n^T})(\bd{Q}^{n^T}(n,:))^t\big)=0$ $\forall n\in\mathcal{N}^{n^T}\setminus n_0^{n^T}$. (16) can be expressed as a function of $\bd{z}$: $\alpha_m {\bd{P}^{n^T}}(i,j)+\beta_m {\bd{Q}^{n^T}}(i,j)+\delta_m {\bd{S}^{n^T\max}}(i,j) \leq 0, \forall m \in \mc{M}\,\, \forall \{i,j\}\in\mathcal{L}^{n^T}$. Finally, (17)-(21) impose lower and upper-bounds on $\bd{x}^{n^T}$ and $\bd{z}$ components.

%\appendices
%\section{Proof of the First Zonklar Equation} \label{appendix}

%\section{}

%\section*{Acknowledgment}


\ifCLASSOPTIONcaptionsoff
  \newpage
\fi

\end{document}